\documentclass[12pt]{article}
\usepackage[latin9]{inputenc}
\usepackage{geometry}
\usepackage{verbatim}
\usepackage{textcomp}
\usepackage{amsmath}
\usepackage{amssymb}
\usepackage{setspace}
\makeatletter

\DeclareTextSymbolDefault{\textquotedbl}{T1}

\usepackage{setspace}
\usepackage{tikz}
\usepackage{framed}
\usepackage{url}
\usepackage{enumerate}
\usepackage{paralist}
\usepackage{verbatim}
\usepackage{caption}
\usetikzlibrary{calc, decorations.pathreplacing,positioning, arrows.meta}
\newtheorem{proposition}{Proposition}

\newtheorem{assumption}{Assumption}
\newenvironment{proof}[1][Proof]{\noindent\textbf{#1.} }{\ \rule{0.5em}{0.5em}}

\makeatother

\begin{document}
\title{\textbf{Optimal Design of Limited Partnership Agreements}}
\date{\today}
\author{Mohammad Abbas Rezaei\thanks{Haas School of Business at UC Berkeley. Email address: mohammad\_rezaei@haas.berkeley.edu}}
\maketitle
\begin{abstract}
\begin{singlespace}
\noindent 
General partners (GP) are sometimes paid on a deal-by-deal basis and
other times on a whole-portfolio basis. When is one method of payment
better than the other? %
{} I show that when assets (projects or firms) are highly correlated
or when GPs have low reputation, whole-portfolio contracting is superior
to deal-by-deal contracting. In this case, by bundling payouts together,
whole- portfolio contracting enhances incentives for GPs to exert
effort. Therefore, it is better suited to alleviate the moral hazard
problem which is stronger than the adverse selection problem in the
case of high correlation of assets or low reputation of GPs. In contrast,
for low correlation of assets or high reputation of GPs, information
asymmetry concerns dominate and deal-by-deal contracts become optimal,
as they can efficiently weed out bad projects one by one. These results
shed light on recent empirical findings on the relationship between
investors and venture capitalists.
\end{singlespace}

\end{abstract}
\thispagestyle{empty}

\noindent \onehalfspacing

\newpage{}

\addtocounter{page}{-1}

\section{Introduction}

In private equity, the relationship between fund managers (general
partners or GPs) and investors (limited partners or LPs) is governed
by a \textquotedblleft limited partnership agreement\textquotedblright{}
(LPA). These contracts are crucial in determining how GPs behave for
the following reasons. First, LPs have limited resources outside of
these contracts to discipline GPs. Second, these agreements typically
remain in effect for about a decade, and recently up to 15 years,
(with little room for renegotiation). Finally, GPs' actions are hard
to observe and writing a contract which provides the right incentives
for GPs is of critical importance. \\
\\
In general, there are three main financial components in an LPA. These
are the management fee, carried interest, and the method of payments
to GPs. While the structure of the management fee and carried interest
has been the subject of extensive research, there is virtually no
theory on why the method of payment is important and how it effects
GPs' performance. Historically LPAs offer two methods for paying carried
interest to GPs. The first method is deal-by-deal or \textquotedblleft American\textquotedblright .
This provision allows GPs to earn the interest as soon as each deal
is exited. The second method is whole-fund or \textquotedblleft European\textquotedblright .
In this method, LPs receive the entire interest on their investment(s)
before GPs get any carried interest.\footnote{In Litvak (2009), there is detailed explanation on different provisions
for these methods.} \\
\\
At first glance, it seems that the European method is more favorable
to LPs-{}-in fact, Huther et al (2020) calls it the \textquotedbl LP-friendly
contract\textquotedbl . In particular, if we assume that the GP does
not change her strategy under different types of contracting, then
whole-portfolio contracting is preferred to the deal-by-deal method
for investors. However, as the GP changes her strategy as the contract
changes, it is not clear which method is more efficient for investors.
Here is an example which illuminates the difference between these
two methods. Suppose a GP has invested in a fund consisting of two
firms. Suppose one of the firms exits with a high return but the other
one loses money so that in total the return is low. In a deal-by-deal
contract, the GP would get some interest on the successful exit. However
in the whole-portfolio method, since the low-return investment offsets
the high-return one, the GP will receive almost nothing and the whole
return will go to the investor. \footnote{Even in the presence of claw-back provisions which requires GPs to
return some of the return at the end to the LP, still the GP gets
an interest free loan from the LP in the meantime. Moreover in the
sample of contracts considered in Cumming and Johan (2009), only about
26\% of contracts had claw-back. }\\
 \\
To fix ideas, consider the following scenario. Suppose there is an
LP who wants to invest in a pool of two projects but has no expertise
to find profitable investment opportunities. As a result, he hires
a GP to do the job. The GP has to exert effort to find good investment
opportunities, but even with significant effort she may end up with
low quality projects. As is prevalent in this setting, the LP has
no control over the GP's actions, nor does he know the quality of
the projects, unlike the GP. Thus, the contracting is subject to both
moral hazard and adverse selection.\\
\\
Within the setting outlined above, I investigate the conditions under
which each method of payment (deal-by-deal or whole-portfolio) is
optimal. As a result, I can explain some empirical findings documented
in the literature. First, I show that when projects are highly correlated,
whole-portfolio contracting is optimal for the LP. As the correlation
declines, the space of portfolios where deal-by-deal contracting is
preferred expands. This phenomena has been documented empirically
in Magro (2018). The mechanism behind this result comes from the trade-off
between the moral hazard about the effort to find good projects versus
the information asymmetry about the quality of projects. When projects
have high correlation, bundling the performance of projects together
can enhance incentives for the GP to exert effort on them. In this
case, even when projects are subject to different degrees of adverse
selection, the loss of efficiency is still low enough that whole-portfolio
contracting is preferred to deal-by-deal contracting which can handle
adverse selection efficiently. \\
\\
Second I show that when the GP is not reputable, it is more likely
that the LP should use whole-portfolio contracting compared to when
contracting with a reputable agent. \footnote{By non-reputable agent, I mean an agent that investor can not verify
her access to investment opportunities hence needs to be distinguished
from fly-by-night operators.}This result is in alignment with the findings in Huther et al (2020).
In this paper, the authors propose that when a GP is more reputable,
they have more market power and hence can get more favorable contracting
terms. In my setting, however, this comes from the fact that for non-reputable
agents, whole-portfolio contracting can reduce the chance of making
bad investments, hence improving the investment strategy. Therefore,
the sorting effect exists in this environment, but indirectly as a
result of the change of behavior of the agent due to the terms of
the contract. \\
\\
The model yields other results and predictions. For example, I show
that when there is little or no information asymmetry about the quality
of projects between investor and agent, whole-portfolio contracting
is the dominant form of contracting. This can explain why we see this
form of contracting when the underlying assets are public firms. Specifically
in the case of hedge funds or mutual funds, the payout to the agent
is almost always a function of the performance of the whole portfolio
rather than the individual performance of assets in the portfolio.
I also predict that investors' information can affect the method of
payment. When investors are not fully informed on the structure of
an investment, they prefer to have a narrower scope of investment
(hence higher correlation) and use the whole-portfolio contracting
method. \\
\\
The main feature of the model which enables me to show these results
is the fact that projects are heterogeneous. If different projects
are always subject to same degree of moral hazard and information
asymmetry, then bundling the payouts together has no efficiency loss
and whole-portfolio contracting is the dominant method of contracting,
as is the case for many contracts in the real world. This is the dominant
assumption in the literature, in the seminal work of Diamond (1984)
and subsequent studies. For example, Laux (2001) considers a pool
of homogeneous projects and show how investors can design better contracts
by the pooling and loosening of limited liability. However, when a
typical VC invests in a pool of projects, it is reasonable to assume
a high degree of heterogeneity between projects.\footnote{VCs invest in projects which are highly innovative with unique business
plans with very few assets in place unlike for example banks which
give loan to ordinary businesses or mortgages to residential/commercial
properties. As a result, we expect much more heterogeneity in VCs
invested portfolios. } \\
\\
The heterogeneity of projects creates a trade off between moral hazard
and adverse selection. When a contract is written on the whole-portfolio
basis, investors can more easily persuade agents to exert effort on
the projects through bundling the payouts. However, whole-portfolio
contracting takes away the flexibility to deal with the different
degrees of adverse selection that the projects are subject to. For
higher correlation between projects or lower reputation of agents,
the priority is to mitigate the more severe moral hazard problem,
and the whole-portfolio contracting is therefore preferred. On the
contrary, when the correlation between projects is low or the agent
is reputable, adverse selection is more severe, and deal-by-deal contracting
is better suited to deal with this issue.\\
\\
This paper relates to the theoretical literature in the area of PE
funding. In Axelson et al (2009), the authors study the problem of
leverage in buyouts and show that a combination of ex-ante pooled
financing and ex-post deal-by-deal financing is optimal. In their
setting, the timing of the investment on projects is different, while
in a lot of limited partnership contracts the GP is required to choose
the portfolio firms early in the life-span of the LPA. In another
similar work, Fang (2019) shows why LPs restrict the investment timing
of GPs. In both of these works, the authors abstract away from the
moral hazard problem between LPs and GPs, and also consider a pool
of similar projects. Because of the homogeneity between projects,
when the method of financing is restricted to ex-ante, whole-portfolio
financing is always optimal in their setting and they are not able
to explain the abundance of the deal-by-deal ex-ante contracting in
the PE industry. \\
\\
This work also contributes to the literature on investment pooling
and portfolio contracting. Inderst et al (2007) consider the case
in which investors faces multiple agents and investment pooling and
credit rationing can motivate optimal investment strategy. Their main
mechanism relies on the competition among agents, while in my work
credit rationing has no bite as investors face only one agent. Fulghieri
and Sevilir (2009) also consider the case of contracting between an
investor and multiple agents and focus on the double moral hazard
problem between GPs and entrepreneurs. In contrast, I abstract away
from GP/entrepreneur problems and focus on the contracting between
GPs and LPs. This paper also relates to the literature on moral hazard
with learning, He et al (2017) and Miao and Rivera (2016), and experimentation
and Bandit problems, Pourbabaee (2020), as well.\\
\\
Empirically, the first work which addresses the importance of the
method of compensation in VC settings is Litvak (2009). She shows
that the shift in the timing of compensation can affect the present
value of the payment to the GP as much as changing the contracting
terms themselves. While the importance of the compensation method
is discussed in Litvak (2009), Huther et al (2020) and Magro (2018)
study the effects of payment methods on the GP investment strategy
and fund's return. All of these papers are empirical and offer little
theory on the matter.\\
\\
More broadly, the first work which studies GP compensation is Gompers
and Lerner (1999). The authors explore the cross sectional and time
variation of the management fee and carried interest in the contract
terms, assuming that contracts have the same method of payment. Metrick
and Yasuda (2010) study a similar problem using an option-pricing
framework, and focus more on buyout funds. Unlike these works, Robinson
and Sensoy (2013) have access to cash flow data as well as contracting
terms, which links the management payment to the performance using
a novel data set containing all the payment from a big institutional
investor to GPs.\\
\\
The paper proceeds as follows. Section \ref{sec:Model}, introduces
the models and shows the optimal contracting on one project. In Section
\ref{sec:Optimal-Portfolio-Contracting}, I solves the problem of
optimal whole-portfolio contracting and compare it to the deal-by-deal
contract. Section \ref{sec:Non-Reputable-GP} consider the same problem
for non-reputable GPs and I compare the results to the case of reputable
agents. Section \ref{sec:Extension} considers various extensions
of the model. Finally Section \ref{sec:Conclusion} concludes.

\section{Model\label{sec:Model}}

There are three classes of agents in the model: limited partners (investors
or LPs), general partners (GPs) and fly-by-night operators (FNOs).
All agents are risk-neutral and have access to a safe asset technology
with a return which is normalized to zero. There are two types of
general partners, reputable and non-reputable. Both types of general
partners have access to a pool of projects in which they can invest
in. The limited partner has capital which is needed to run projects.\footnote{Throughout the paper, I use he/him to refer to the LP and she/her
to refer to the GP. } FNOs have no access to the pool of risky projects but they can mimic
the behavior of a GP. I assume that if the GP is reputable, then the
LP can verify that she has access to projects. However, the LP can
not distinguish between a non-reputable agent and a FNO. Initially,
I focus on reputable GPs and discuss contracting with non-reputable
GPs in Section \ref{sec:Non-Reputable-GP}. Every project needs an
investment outlay of $I$. The GP has no initial money and should
raise it from the LP if she decides to invest in the project(s).\footnote{The assumption that GP has no initial capitl has no effect on the
results. We can assume that GP needs extra capital $I$ as long the
payout of the contract to the GP is at least as her initial capital. } There are two types of projects, $\theta\in\{G,B\}.$ A good project
(type $G$) has guaranteed return $R$ (hence it is always successful)
but a bad project (type $B$) has return $R$ with probability $p$
and return 0 with probability $1-p.$ The GP can also opt to not invest
in a project and invest the raised capital in a safe asset, therefore
receiving the return $I$. Hence, the possible outcomes are $\{0,I,R,2I,R+I,2R$\}
if the GP raises enough capital for two projects ($2I)$. Clearly,
if the GP raises only $I$, then possible returns are $\{0,I,R\}.$
Type $B$ projects are negative NPV, so I assume

\begin{assumption}\label{assumption-1}

\[
pR<I.
\]

\end{assumption}

The GP can exert effort to increase the chance of getting a good project.
Moreover, if the GP exerts no effort for a project, then the project
which is chosen is guaranteed to be bad (type $B$). Otherwise, if
the GP exerts a binary effort with cost $c$, the chance of getting
a good project (type $G$) is $\lambda$. I assume that the decision
to exert effort is optimal in the following sense 

\begin{assumption}\label{assumption-2}

\[
\lambda R+(1-\lambda)I>I+c
\]

\end{assumption}

which can be written as 

\begin{equation}
R-\frac{c}{\lambda}>I.\label{eq:feasible-FB}
\end{equation}
However, it is possible that the agent exerts effort but does not
commit to not invest if the quality is bad. In this case, the return
is

\[
\lambda R+(1-\lambda)pR
\]
 which is less than $\lambda R+(1-\lambda)I$ by assumption \ref{assumption-1}.
It can be seen that if there is no agency friction, then when equation
(\ref{eq:feasible-FB}) holds, the agent/investor exerts effort to
obtain a good project and invest in the project if he ends up with
a type $G$, otherwise keeping the money in a safe asset. In this
case, the profit made from the project is $\lambda(R-I)-c.$ In my
definition $\frac{1}{\lambda}$ measures the extent of the moral hazard
issue. Higher $\lambda$ means higher chance of obtaining a good project,
so the moral hazard problem is less severe. On the other hand, $p$
measures the extent of the information asymmetry between agents, since
for higher $p$ it is harder to give incentives for the GP to not
invest in a bad project. 

The model has three dates $t=0,1,2$ and two periods. At $t=0,$ the
contract between the GP and the LP is written and capital is raised.
Then between dates 0 and 1, the GP can exert effort to increase the
chance of getting a good project. At $t=1,$ the type of projects
is revealed to the GP and she makes the investment either in these
projects or safe assets. Finally, at $t=2$, cash flows are realized
and agents receive their money based on the contract. In the real
world, it is possible that projects exit at different times but if
the contract is based on whole-portfolio performance, then money is
stored in an escrow account until distributed later when all the projects
exit. Hence my assumption on having the same exit time is not unrealistic. 

\subsection{Deal-by-deal contract\label{subsec:Deal-by-deal-contract}}

In this section, I consider the contracting problem when the contract
between the GP and the LP is written in a deal-by-deal way. Since
agents are risk-neutral, the optimal deal-by-deal contract consists
of two optimal contracts on a single project. Therefore, I only need
to study the contracting problem for one project. 

In order to fund projects, claims $s_{GP}(x)=s(x)$ and $s_{LP}(x)=x-s_{GP}(x)$
are issued, which determines how much agents will receive when the
payout of the project is $x.$ I impose following a priori assumptions
on the payout of securities. 
\begin{itemize}
\item \textbf{Limited Liability:\;}$0\leq s_{GP}(x),s_{LP}(x)$.
\item \textbf{Monotonicity:\;}$s(x)$ and $x-s(x)$ are non-decreasing
in $x$ when $x$ is an outcome on the equilibrium path.
\end{itemize}
This monotonicity assumption is common in the literature on security
design--See Nachman and Noe (1994), for example. Sometimes the security
is assumed to be monotonic on the whole possible set of payouts. .
I will revisit this issue in Section \ref{sec:Appendix-B-:Robustness}.

The LP can not observe the quality of the chosen project (projects)
or if the GP exerts effort or not. However, the LP can observe whether
the GP invests in the project or in the safe asset. Also, the cash
flow is verifiable at the end of period 2 as well. 

For one project, after the issuance of $s(x)$, there are four possible
strategies by the GP.
\begin{enumerate}
\item Do not invest: The return to GP is $s(I).$
\item Invest with no effort: The return is $ps(R).$
\item Exert effort and invest regardless of quality: $\lambda s(R)+(1-\lambda)ps(R)-c$.
\item Exert effort and invest only in the good project: $\lambda s(R)+(1-\lambda)s(I)-c.$
\end{enumerate}
Clearly, the optimal strategy is the fourth one if there was no agency
friction (I will show later that strategy (4) is also optimal even
in the presence of agency friction.) . Assuming this, here is how
the LP can implement strategy (4). The scheduled payment's system
$(s(I),s(R))$ induce the GP to choose strategy (4) if and only if
it satisfies 
\begin{align}
 & s(R)\geq s(I)+\frac{c}{\lambda}\label{eq:opt-dealbydeal1}\\
 & s(I)\geq ps(R).\label{eq:opt-dealbydeal2}
\end{align}

The first condition insures that the agent exerts effort to obtain
a good project and the second one insures that the agent does not
invest if the quality of the project turns out to be bad. As usual,
we have the participation constraint by the LP which is 
\begin{equation}
E[s_{LP}(x)]=E[x-s(x)]\geq I.\label{eq:PC-DBD}
\end{equation}
The problem faced by the LP can then be written as 
\[
\max_{s(I),s(R)}E[x-s(x)]
\]
where $(s(I),s(R))$ satisfy equations (\ref{eq:opt-dealbydeal1})
and (\ref{eq:opt-dealbydeal2}). Inserting equation (\ref{eq:opt-dealbydeal2})
in equation (\ref{eq:opt-dealbydeal1}) gives
\[
s(R)\geq ps(R)+\frac{c}{\lambda}
\]
hence $s(R)\geq\frac{c}{\lambda(1-p)}$. As a result, the optimal
contract will be $(\frac{c}{\lambda(1-p)},\frac{pc}{(1-p)\lambda}$).
Using equation (\ref{eq:PC-DBD}), The project is funded if and only
if we have 
\begin{equation}
\lambda[R-\frac{c}{\lambda(1-p)}]+(1-\lambda)[I-\frac{pc}{(1-p)\lambda}]\geq I.\label{eq:feasible-No-FNO}
\end{equation}
Under this contract, the profit made by the GP is 
\begin{align}
 & \Pi_{GP}=\Pi=\lambda\frac{c}{\lambda(1-p)}+(1-\lambda)\frac{pc}{\lambda(1-p)}-c\nonumber \\
 & =\frac{pc}{\lambda(1-p)}\label{eq:Profit-GP}
\end{align}
Since the optimal effort/investment strategy is chosen by the GP,
we have $\Pi_{LP}+\Pi_{GP}=\lambda(R-I)-c$. So the profit made by
the LP from the contract on one project is 
\begin{equation}
\Pi_{LP}=\lambda(R-I)-c-\frac{pc}{\lambda(1-p)}.\label{eq:LP-profit-DbD}
\end{equation}
We can see that for higher $p$, the LP makes less profit (and therefore
the GP makes more). This is because, as mentioned before, higher $p$
is associated with more severe adverse selection and it makes it harder
to motivate the GP to invest optimally since the outside option (the
bad project) is more appealing. On the contrary, when $\lambda$ goes
up, the profit goes up for the LP (and down for the GP) because the
chance of success when exerting effort is higher, so less payment
is needed to motivate effort. It is also worth noting that when there
is no bad option for investment by the GP, which means $p=0,$ then
the LP can get the whole surplus of the project. This case is effectively
means that there is no asymmetric information between the LP and the
GP. As a result, in a setting with binary effort, contracting alleviates
all the friction in the model. In section \ref{sec:No-Asymmetric-Information},
I consider implications when there is more variance for effort in
this special important case. 

\begin{figure}
\centering
\begin{tikzpicture}[x=0.65cm, font=\small]

\coordinate (start) at (-5,0);
\coordinate (end) at (9,0);
\draw [line width=2pt] (start) -- (end);

\coordinate (s0) at (-4,0);
\coordinate (t0) at ($(s0)+(0,0.3)$);
\coordinate (s1) at (2,0);
\coordinate (t1) at ($(s1)+(0,0.3)$);
\coordinate (s2) at (8,0);
\coordinate (t2) at ($(s2)+(0,0.3)$);
\coordinate (s'0) at ($(s0)+(0,-0.2)$);

\draw [black, ultra thick ,decorate,decoration={brace,amplitude=5pt},
       xshift=0pt,yshift=-4pt] (-3,0.5)  -- (1,0.5) 
       node [black,midway,above=4pt,xshift=-2pt] { \tiny GP exerts effort.};
\draw [line width=2pt] (s0) -- (t0);
\node [anchor=south] at (t0.north) {$t=0$};

\draw [line width=2pt] (t1) -- (s1);
\node [anchor=south] at (t1.north) {$t=1$};

\draw [line width=2pt] (t2) -- (s2);
\node [anchor=south] at (t2.north) {$t=2$};

\node [anchor=north, align=left, text width=5cm] at (s'0.south) {
\begin{itemize}
\item Securities $s_{GP}$ and $s_{LP}$ \\ are issued.
\item $2I$ is raised.
\end{itemize}
};

\node [anchor=north, align=left, text width=5cm] at (s1.south) {
\begin{itemize}
\item GP realises the type of projects.
\item Investment decision is \\ made.
\end{itemize}
};

\node [anchor=north, align=left, text width=5cm] at (s2.south) {
\begin{itemize}
\item Cash flows are realized.
\item It is distributed according \\ to securities.
\end{itemize}
};

\end{tikzpicture}
\caption{Timeline}
\end{figure}

In order to compare the outcome of the results of different strategies
induced by the LP, first note that the LP never induces strategy (1)
as he personally has access to the safe asset. The second strategy
has always negative NPV since the profit by the LP is 
\begin{align*}
 & p(R-s(R))-I\\
 & \le pR-I<0
\end{align*}
 by assumption \ref{assumption-1}. Finally to optimally induce strategy
(3), note that in this case the LP optimally sets $s(I)=0$ as $I$
is not the outcome of the induced strategy. To induce effort, the
payout should satisfy 
\[
s(R)\geq ps(R)+\frac{c}{\lambda}
\]
 which implies $s(R)\geq\frac{c}{\lambda(1-p)}$. Hence the LP issues
security $(0,\frac{c}{\lambda(1-p)}).$ The profit made by the LP
by this contract is 
\[
(\lambda+(1-\lambda)p)(R-\frac{c}{\lambda(1-p)})-I.
\]
This is less than the profit made by the LP by strategy (4) (equation
(\ref{eq:LP-profit-DbD})) by assumption \ref{assumption-1}. 

In summary, we have the following for investment on one project.

\begin{proposition}[optimal deal-by-deal contract]\label{deal-by-deal}

The optimal strategy that the LP induces the GP to choose is strategy
(4). Moreover, the security $(s(I),s(R))=(\frac{c}{\lambda(1-p)},\frac{pc}{(1-p)\lambda})$
is issued optimally by the LP to fund the project. The funding is
possible if and only if 
\begin{equation}
\Pi_{LP}=\lambda(R-I)-c-\frac{pc}{\lambda(1-p)}\geq0.\label{eq:profit-LP}
\end{equation}

\end{proposition}

When there are two projects with parameters $(\lambda_{1},p_{1})$
and $(\lambda_{2},p_{2}),$in a deal-by-deal contract, the optimal
contract for one project is written for each of the projects. Therefore
the expected profit by the GP will be 
\[
\Pi_{GP}=\sum_{i=1}^{2}\frac{p_{i}c}{\lambda_{i}(1-p_{i})}.
\]
In the next section, we see how tying the payouts of the projects
together can change the expected payout to the LP (and the GP).

\section{Optimal Portfolio Contracting\label{sec:Optimal-Portfolio-Contracting}}

In this section, I analyze the question of optimal contracting when
a portfolio of projects is chosen by the GP and the payout can depend
on the whole return of the portfolio. Then I compare whole-portfolio
contracting with deal-by-deal contracting to see how the investment
environment can affect the choice of contract by investors. But first
I need to introduce the dependency between projects, which I do in
the next part. 

\subsection{Correlation structure of the portfolio \label{subsec:Correlation-structure}}

When the GP forms a portfolio of investments, not only the return
and quality of each project is important, but the correlation structure
between projects is important as well. In Magro (2018), the correlation
structure in the invested portfolio under different types of contracting
is studied, analyzing the correlation between projects in two dimensions
of industry and geography. 

Here I assume that the correlation between projects is given by a
parameter $0\leq\rho\leq1$ which means that if the GP exerts effort
on both projects, then 
\begin{equation}
\mathbb{P}[\textrm{project\;1\;and\;2\;are\;good\;\ensuremath{\|\;}effort exerted on both projects]=}\rho\min(\lambda_{1},\lambda_{2}).\label{eq:correlation-relation}
\end{equation}
Here I assume that investors can observe $\rho$ and potentially write
a contract conditioned on it. I relax this condition in section \ref{sec:Uninformed-Investor}.The
cost of exerting effort on both projects is twice that of one project,
$2c$. The correlation structure is irrelevant when the contract is
written on a deal-by-deal basis. This is because, by risk neutrality,
deal-by-deal contracting is equivalent to writing a contract with
two different agents. Therefore I only need to study the problem when
the payout depends on the payout of both projects. 

Before going into detail on the whole-portfolio contracting, let me
introduce some preliminary results which are needed later in the discussion.
As in the deal-by-deal case, the GP should not invest in the type
$B$ project. Hence possible optimal outcomes from the projects are
$2I,R+I$ and $2R.$ These correspond to cases in which the GP comes
up with zero, one or two good projects respectively. As a result,
$I$ and $R$ are not possible outcomes if the GP makes optimal investment
decisions. Hence to minimize the incentive for these outcomes, it
is easy to show that the optimal contract satisfies 
\begin{equation}
s(0)=s(I)=s(R)=0.\label{eq:s(I)=00003Ds(R)=00003D0}
\end{equation}

\subsection{Whole-Portfolio Contract \label{subsec:Whole-Portfolio-Contract}}

In this section, I want to see how the optimal contract should be
written when the return is a function of the total payouts of the
projects. This resembles whole-portfolio contracting. I then compare
it to deal-by-deal contracting to find under which parameters each
type of contracting is efficient . 

When writing the contract on the whole portfolio, as we saw in equation
(\ref{eq:s(I)=00003Ds(R)=00003D0}), we have $s(I)=s(R)=s(0).$ Set
\begin{equation}
(x,y,z)=(s(2R),s(R+I),s(2I))\label{eq:payout-abr}
\end{equation}
The LP needs to impose some restrictions on the payment to the GP
to make sure that the GP only invests in good projects. These conditions
are 
\begin{align*}
 & z\geq\max\{p_{1}y,p_{2}y,p_{1}p_{2}x\}\\
 & y\geq\max\{z,p_{1}x,p_{2}x\}\\
 & x\geq\max\{z,y\}.
\end{align*}
These inequalities make sure that the GP will invest in good projects
and only in good projects (hence withholding money from bad projects).
For example, when the agent comes up with two bad projects, the payout
for not investing in any bad project ($z$) is not less than the (expected)
payout if the agent invests in one bad project ($p_{i}y$ for $i=1,2$)
or invests in two bad projects ($p_{1}p_{2}x$). A similar explanation
applies to $y_{i}\geq p_{3-i}x$ for $i=1,2$. These conditions therefore
discourage the GP from making bad investment decisions. Also since
$x\geq y\geq z$, the GP will invest in good projects when they are
available rather than investing in the safe asset. Note that when
the contract satisfies these conditions, the payout to the GP is increasing
on the equilibrium as the GP does not invest in a bad project, hence
satisfying the monotonicity condition. Set $\lambda_{max}=\max[\lambda_{1},\lambda_{2}]$
and $\lambda_{min}=\min[\lambda_{1},\lambda_{2}]$. In addition, the
GP should have incentive to exert effort on both projects. This gives
\begin{align*}
 & \rho\lambda_{min}x+[\lambda_{1}-2\rho\lambda_{min}+\lambda_{2}]y+[1-\lambda_{1}-\lambda_{2}+\rho\lambda_{min}]z\\
 & \geq\max\{z+2c,\lambda_{max}y+(1-\lambda_{max})z+c\}
\end{align*}
The LHS term is the expected payout to the GP if she exerts effort
on both projects. On the RHS we have expected payouts if no effort
is exerted or if it is exerted on only one project. For simplicity
and without loss of generality, assume that $p_{1}\geq p_{2}.$ Note
that in the optimum
\begin{align}
 & \rho\lambda_{min}x+[\lambda_{1}-2\rho\lambda_{min}+\lambda_{2}]y+[1-\lambda_{1}-\lambda_{2}+\rho\lambda_{min}]z\label{eq:optimum-MH}\\
 & =\max\{z+2c,\lambda_{max}y+(1-\lambda_{max})z+c\}.\nonumber 
\end{align}
\tikzstyle{level 1}=[level distance=3.5cm, sibling distance=3.5cm]
\tikzstyle{level 2}=[level distance=3.5cm, sibling distance=2cm]

\tikzstyle{bag} = [text width=4em, text centered]
\tikzstyle{end} = [circle, minimum width=3pt,fill, inner sep=0pt]

\begin{figure}
	\centering
\begin{tikzpicture}[grow=right, sloped]
\node[bag] {GP}
    child {
        node[bag] {Quality $\theta_2$ Project}        
            child {
                node[end, label=right:
                    {Invest in the safe asset}] {}
                edge from parent
            }
            child {
                node[end, label=right:
                    {Invest in the project 2}] {}
                edge from parent
            }
            edge from parent 
            node[below]  {effort $e_2$}
    }
    child {
        node[bag] {Quality $\theta_1$ Project}        
        child {
                node[end, label=right:
                    {Invest in the safe asset}] {}
                edge from parent
            }
            child {
                node[end, label=right:
                    {Invest in the project 1}] {}
                edge from parent
            }
        edge from parent         
            node[above] {effort $e_1$}
    };
\end{tikzpicture}
\caption{GP Problem}
\end{figure}
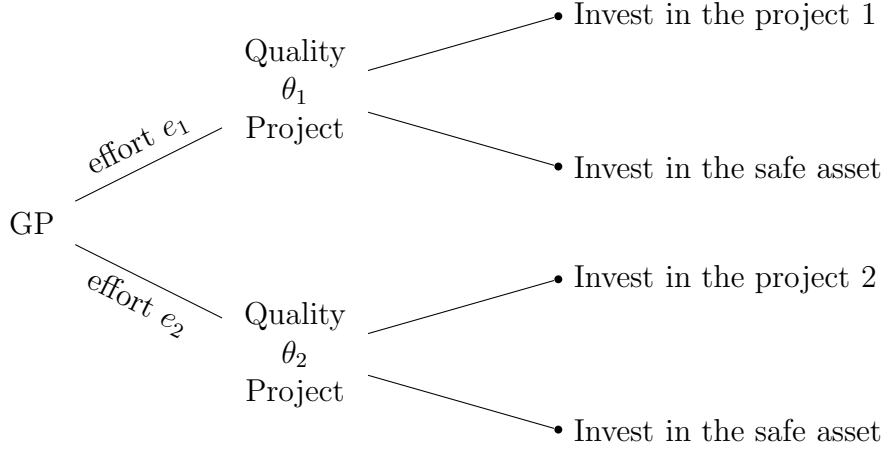

This is a typical phenomenon when dealing with moral hazard issues.
This means that the expected payout to the GP is binding by the condition
which induces exerting effort on both projects, otherwise the investor
can lower payment in some states of the world without changing GP
incentives. More formally, if equation (\ref{eq:optimum-MH}) does
not hold, the transformation $z\rightarrow z-\epsilon$ for small
enough $\epsilon$ should violate the optimality conditions. Otherwise
the LP can have a feasible contract with less expected payout to the
GP. This means that $z=p_{1}y,$ hence $z$ can not be reduced. Similarly
$y\rightarrow y-\epsilon$ should violate the conditions as well,
hence one gets $y=p_{1}x$. But then in any case $x\rightarrow x-\epsilon$
is possible because by the last two equalities we have $x>y\geq z.$
Therefore the LP problem can be written as 
\begin{align*}
 & \min_{x,y,z}\alpha x+\beta y+\gamma z\\
 & \alpha x+\beta y+\gamma z=\max\{z+2c,\lambda_{max}y+(1-\lambda_{max})z+c\}\\
 & \ensuremath{x\geq y\geq z\geq p_{1}y\geq p_{1}^{2}x}
\end{align*}
where $(\alpha,\beta,\gamma)=(\rho\lambda_{min},\lambda_{1}-2\rho\lambda_{min}+\lambda_{2},1-\lambda_{1}-\lambda_{2}+\rho\lambda_{min})$.
In a similar vein, conditions that induce the choice of optimal investment
strategies by the GP are binding as well. Therefore we have the following
proposition.

\begin{proposition}[Optimal whole-portfolio contract]\label{prop: Optimal-whole-portfolio}

The optimal whole-portfolio contract satisfies 
\begin{equation}
z=p_{1}y=p_{1}^{2}x\label{eq:optimal-whole}
\end{equation}
 where 
\[
z=\frac{2c}{\beta\frac{1-p_{1}}{p_{1}}+\alpha\frac{1-p_{1}^{2}}{p_{1}^{2}}}
\]
 if 
\[
\rho\geq\rho^{*}=\frac{\lambda_{max}-\lambda_{min}}{\lambda_{min}(\frac{1}{p_{1}}-1)}
\]
 otherwise 
\[
z=\frac{c}{\frac{\alpha(1-p_{1}^{2})}{p_{1}^{2}}+\beta\frac{1-p_{1}}{p_{1}}-\lambda_{max}\frac{1-p_{1}}{p_{1}}}
\]

\end{proposition}

Here is the intuition behind this proposition. When the investor writes
the contract, he wants to choose the maximal value for $x$ to give
the biggest incentive to the agent to exert effort. However, because
of adverse selection, the prize for success cannot be too large, as
it leads to inefficient investment decisions by the GP (investing
in type $B$ projects). The maximum possible value for $x$ is $\frac{z}{p_{1}p_{2}}$
and maximum attains if and only if equation \ref{eq:optimal-whole}
holds. The two different regimes in the proposition correspond to
the fact that the expected payout to the GP ($\alpha x+\beta y+\gamma z)$,
becomes equal to $z+2c$ or $\lambda_{max}y+(1-\lambda_{max})z+c.$
When the correlation is high, the GP either prefers to exert no effort
or to exert effort on both projects because of high dependency between
the success in both projects. Hence for high values of $\rho,$it
is needed to pay enough to the GP such that the GP exerts any effort
at all. This amount is $s(2I)$ which is the reserve value for the
GP. However, when the correlation is low, the payout should compensate
for the lower level of inter-dependency between projects. So the payout
should be high enough for the case success in both projects so that
the GP does not find it beneficial to exert effort on only the easier
project (corresponding to $\lambda_{max})$. 

Note that, as we mentioned, the security as defined here is increasing
on the set of possible outcomes on the equilibrium path. However since
$s_{GP}(R)=0$ and $s_{GP}(2I)>0$, when $R>2I$, optimal security
is not increasing on all possible outcomes. This stems from the fact
that the LP wants to push the GP to invest in only good projects and
reserve the money if the project is bad. The non-monotonicity of the
optimal security has been observed before in the literature, like
Manso (2011). The mechanism in Manso (2011) which leads to this phenomena
is the fact that the contract is written in a way to motivate experimentation
by the agent. Hence the principal has to reward for failure so that
the agent can take the risk. However, my setting has quite an opposite
mechanism--the monotonicity arises because the GP wants to make the
LP take less risk. For example, suppose the principal wants the agent
to invest in a safe asset, and the agent goes and invests in a bad
project instead. So if the payoff is high, it means that the agent
deviated from the optimal strategy and as a result she gets punished. 

The derivative of $\alpha\frac{1+p_{1}}{p_{1}}+\beta$ with respect
to $\rho$ equals to 
\[
\lambda_{min}(\frac{1}{p_{1}}-1)>0
\]
 hence the payout to the GP decreases as $\rho$ increases. Also note
that the total payout of projects (which is $E[s_{GP}]+E[s_{LP}]$)
equals to 
\[
\alpha2R+\beta(R+I)+\gamma2I
\]
 it has derivative (w.r.t $\rho$)
\[
2\lambda_{min}R-2\lambda_{min}(R+I)+2\lambda_{min}I=0
\]
so total payout of projects is constant. Therefore we have the following 

\begin{proposition}[Comparing whole-portfolio with deal-by-deal]\label{prop: WP-VS-DBD-Reputable}

As the correlation $\rho$ increases, the expected payout to the GP
decreases and the expected payout to the LP increases. Therefore,
for admissible values ($\lambda_{1},\lambda_{2},p_{1},p_{s}),$ there
is $\rho^{**}=\rho(\lambda_{1},\lambda_{2},p_{1},p_{s}),$such that
for $\rho>\rho^{**},$ whole-portfolio contracting is preferred by
the LP and for $\rho<\rho^{**}$deal-by-deal contracting is preferred
by the LP. In addition if $\rho\geq\rho^{*}$, whole-portfolio contracting
is better for the LP ( equivalently $\rho^{**}\leq\rho^{*}$).

\end{proposition}

The intuition for the proposition above comes from the fact that when
the correlation between projects is higher, it becomes easier to encourage
the GP to exert effort on both projects since success in one project
increases the chance of success in the other one. As a result, the
LP needs to pay less to motivate effort by the GP, hence the LP makes
more profit because the total payout of projects is the same for all
$\rho$. Since the deal-by-deal contract is independent from the correlation,
from the monotonicity of payout with respect to correlation, we can
see that if the LP prefers whole-portfolio contracting to deal-by-deal
contracting for a given $\rho,$then as $\rho$ goes up it is still
the case. As I pointed out in the introduction, this result observed
empirically in Magro (2018). There the author shows that deal-by-deal
compensation induces greater heterogeneity in portfolio investments.
So the proposition above rationalizes this finding. In section \ref{sec:Appendix-B-:Robustness},
I show that this result holds when a strong form of monotonicity is
imposed on the security as well.

Whole-portfolio contracting does not depend on $p_{2}.$ This is because
since $p_{1}\geq p_{2},$ the first project has more severe adverse
selection problem compared to the second one. Therefore when information
asymmetry constraint binds for the first project, it is already alleviate
the adverse selection for the second project as well. Mathematically
speaking when $z\geq p_{1}y$ then already we have $z>p_{2}y$ as
well. When $\rho$ is large enough, projects are similar to each other
and as we saw in Proposition \ref{prop: WP-VS-DBD-Reputable}, whole-portfolio
contracting is more appealing for the LP. This is because, in this
case, bundling efforts together gives the LP a big enough benefit
that makes up for the loss which comes from having inefficient treatment
of adverse selection (in contrast to the deal-by-deal contract which
handles this issue efficiently). However for smaller values of $\rho,$
the comparison of benefiting from bundling effort is smaller than
the loss of sub-optimal handling of the information asymmetry problem.
Recall that the expected payout to the GP in the deal-by-deal case
is
\[
2c+\Pi_{GP}=2c+\sum_{i=1}^{2}\frac{p_{i}c}{\lambda_{i}(1-p_{i})}
\]
which is increasing in $p_{2}.$ Therefore the discussion above implies
the following.

\begin{proposition}[Asymmetry of information VS moral hazard]\label{prop:AI-VS-MH-Reputable}

For given values $(\lambda_{1},\lambda_{2},p_{1},\rho)$ , there is
$p_{2}^{*}=p_{2}(\lambda_{1},\lambda_{2},p_{1},\rho)$ such that for
$p_{2}<p_{2}^{*}$ deal-by-deal contracting is better for the LP and
for $(p_{1}\geq)p_{2}>p_{2}^{*}$ whole-portfolio contracting. When
$\rho>\rho^{**},p_{2}^{*}=0.$

\end{proposition}

While Proposition \ref{prop: WP-VS-DBD-Reputable} resolves the comparison
between whole-portfolio contracting and deal-by-deal in terms of correlation,
Proposition \ref{prop:AI-VS-MH-Reputable} helps us to understand
the comparison in terms of information asymmetry. Here is the intuition
behind this statement. As mentioned before, the investor should take
into account the loss of efficient handling of the adverse selection
problem. The term $p_{1}-p_{2}$ measures the difference between the
adverse selection issues that two projects are subject to. When $p_{1}-p_{2}$
is large (equivalently $p_{2}$ is small), the heterogeneity of asymmetry
of information between the two projects is large. As a result, it
is more efficient to have a deal-by-deal contract for better handling
of this issue. Whereas for large $p_{2}$ (small $p_{1}-p_{2}$),
the loss of efficiency on this issue is negligible, hence whole-portfolio
contracting is better. 

Whenever $p_{2}^{*}=0,$ whole-portfolio contracting is dominant for
the set of parameters given. When $\lambda_{max}=\lambda_{1},$then
$p_{2}^{*}>0$ whenever $\rho<\rho^{*}$. However when the moral hazard
problem is more severe in the first project as well (i.e $\lambda_{max}=\lambda_{2}),$then
for a larger set of $\rho$, whole-portfolio contracting is dominant.
In this case, the investor uses the payout on the second project,
which dominates the first in both moral hazard and information asymmetry
aspects, as a prize to motivate GP to exert effort on the first project.
From the proof of the Propositions \ref{prop: WP-VS-DBD-Reputable},
the equation that computes $p_{2}^{*}$ is given by (when $\rho<\rho^{*}$)
\begin{align*}
 & \lambda_{max}\frac{z}{p_{1}}+(1-\lambda_{max})z+c=2c+\sum_{i=1}^{2}\frac{p_{i}c}{\lambda_{i}(1-p_{i})}\\
 & z=\frac{c}{\frac{\alpha(1-p_{1}^{2})}{p_{1}^{2}}+\beta\frac{1-p_{1}}{p_{1}}-\lambda_{max}\frac{1-p_{1}}{p_{1}}}
\end{align*}
Relative to other variables, we have the following Proposition.

\begin{proposition}[Comparative Statics]\label{prop:Comparative-statics}

In the region $\rho>\rho^{*},$expected payout to the GP is increasing
in $p_{1}$ and is decreasing in $\lambda_{max}$ and $\lambda_{min}.$
For $\rho<\rho^{*},$expected payout to the GP is increasing with
respect to $p_{1}$ and $\lambda_{max}$ and decreasing with respect
to $\lambda_{min}$. 

\end{proposition}

Here is the intuition behind Proposition \ref{prop:Comparative-statics}.
In both regimes of $\rho,$ when $p_{1}$ increases, the information
asymmetry to be overcome by the LP worsens as the GP finds it more
profitable to invest in bad projects. As a result, the expected payout
to the GP increases when $p_{1}$ increases to compensate for information
rent by the GP. With respect to $\lambda_{min},$ as $\lambda_{min}$
increases, it becomes easier for the LP to motivate the GP to exert
effort for the harder project, hence the expected payout is decreasing.
However with respect to $\lambda_{max},$ the relation to the payout
depends on which regime $\rho$ is in. In the high correlation regime
($\rho>\rho^{*}),$as we saw above, the GP has to compensate as much
as needed to make the GP exert any effort. As in a classical moral
hazard problem, when the task becomes easier the expected payout to
the agent decreases. However, in the regime $\rho<\rho^{*},$ the
LP has to compensate the GP for the strategy of exerting effort only
on the easier project. This outside option's payout increases as $\lambda_{max}$
increases, hence the LP has to compensate the GP more for not choosing
this strategy.

\section{Non-Reputable GP\label{sec:Non-Reputable-GP}}

Following Axelson et al (2009) and as mentioned in section \ref{sec:Model},
when the GP is not reputable, the LP cannot distinguish the non-reputable
GPs from a fly-by-night operator (FNO). In this case the following
assumption should be imposed on the securities to discourage FNOs
from getting the investment outlay $2I$ and enjoy the managerial
fee $s(2I)$ without exerting any effort..
\begin{itemize}
\item $s_{GP}(x)=0$ for $x\leq K$ where $K$ is the committed capital
(FNO assumption).
\end{itemize}
This assumption was first introduced in Axelson et al (2009) and has
been used in subsequent works (for example Fang (2019)). Here I investigate
how enforcing this condition can change the contract. Therefore, in
essence we have a separating equilibrium in which a contract between
the LP and a non-reputable GP satisfies the FNO assumption while a
contract between the LP and a well-established (i.e., reputable) GP
does not need this condition. 

First note that with this assumption, first-best can not be implemented
for a single project. This happens since $s_{GP}(I)=0$, the GP will
invest in a project no matter the quality of the project as there
is no reward for not investing in a type $B$ project, hence he implements
strategy (3) discussed in section \ref{sec:Model}. The GP exerts
effort if 
\begin{equation}
s^{FNO}(R)\geq\frac{c}{\lambda(1-p)}.\label{eq:opt-FNO}
\end{equation}
It is feasible to fund through this contract if and only if 
\begin{equation}
R-\frac{c}{\lambda(1-p)}\geq\frac{I}{\lambda+(1-\lambda)p}.\label{eq:feasible-FNO}
\end{equation}
Here $\lambda+(1-\lambda)p$ is that chance of having return $R$.
When the project is good (which has probability $\lambda$), return
$R$ has probability 1 and if the project is bad (with probability
$1-\lambda)$, the probability of success is $p$. The profit made
by the GP is 

\begin{align}
 & \Pi_{GP}^{FNO}=(\lambda+(1-\lambda)p)\frac{c}{\lambda(1-p)}-c\nonumber \\
 & =\frac{pc}{\lambda(1-p)}.\label{eq:profit-GP-FNO}
\end{align}

Comparing with equation (\ref{eq:Profit-GP}), we can see that both
reputable and non-reputable agents make the same profit. This comes
from the fact that $s(I)=ps(R)$ for a reputable agent. Therefore
a reputable GP gets the same payout as a non-reputable GP in the case
of getting of a bad project and investing in the safe asset instead
of making a bad investment, which is in the interest of the LP. Not
surprisingly, the feasibility condition (\ref{eq:feasible-FNO}) is
weaker compared to the case of a reputable agent which is (\ref{eq:feasible-No-FNO})
since here the investment strategy by the agent is not optimal. In
the next part I find the optimal whole-portfolio contract and compare
the results with that of section \ref{sec:Optimal-Portfolio-Contracting}.
If $p=0,$ in a single project's contract, LP gets the whole surplus
of the project so for the next part we assume $\max(p_{1},p_{2})>0$.

\subsection{Whole-Portfolio Contracting With FNO\label{subsec:Whole-Portfolio-Contracting-With}}

In this part I find the optimal whole-portfolio contracting in the
presence of the FNO assumption. By the FNO assumption, the contract
satisfies
\[
s^{FNO}(I)=s^{FNO}(2I)=0.
\]
Since $s^{FNO}(2I)=0,$ it is impossible to motivate the GP to not
invest in any bad project when both projects are bad (as the return
to the GP will be zero in this case). So the best strategy that the
LP can hope to achieve, similar to Axelson et al (2009), is the following
\begin{itemize}
\item The GP exerts effort on both projects.
\item If at least one of the projects is good, the GP invests in only good
projects (optimal choice).
\item If both projects are bad, the GP invests in just one bad project. 
\end{itemize}
As we can see, again $R$ is not the outcome of the optimal strategy,
so in the optimal contract $s^{FNO}(R)=0$ holds as well. To induce
(constrained) optimal choice of investment after efforts are exerted,
assuming $p_{1}\geq p_{2},$ the LP should impose 
\[
s^{FNO}(2R)\geq s^{FNO}(R+I)\geq p_{1}s^{FNO}(2R)
\]
 Under these conditions, the GP invests only in good projects if any
are available and invests in only one bad project if both projects
are bad. With this strategy, the total payout to the GP becomes 
\begin{align*}
 & \rho\lambda_{min}s^{FNO}(2R)+(\lambda_{1}-2\rho\lambda_{min}+\lambda_{2})s^{FNO}(R+I)+p_{1}(1-\lambda_{1}-\lambda_{2}+\rho\lambda_{min})s^{FNO}(R+I)\\
 & =\alpha x+\tilde{\beta}y
\end{align*}
 where $(x,y)=(s^{FNO}(2R),s^{FNO}(R+I))$ and $(\alpha,\tilde{\beta})=(\rho\lambda_{min},\lambda_{1}-2\rho\lambda_{min}+\lambda_{2}+p_{1}(1-\lambda_{1}-\lambda_{2}+\rho\lambda_{min}))$.
The only term which is different compared to the reputable agent is
the last term. This comes from the fact that in the case of two bad
projects, the GP invests in the project corresponded to $p_{1}$ and
withhold money on the other one. Similar to what we had in Section
\ref{sec:Optimal-Portfolio-Contracting}, in order to induce effort
on both projects, the contract should satisfy 
\[
\alpha x+\tilde{\beta}y\geq p_{1}y+2c,\theta_{i}y+c
\]
where $\theta_{i}=\lambda_{i}+p_{1}(1-\lambda_{i})$. Therefore, the
LP problem can be written as 
\begin{align*}
 & \min_{x,y}\alpha x+\tilde{\beta}y\\
 & x\geq y\geq p_{1}x;\;\;\;\alpha x+\tilde{\beta}y\geq\max\{\theta_{max}y+c,p_{1}y+2c\}
\end{align*}
 Like before, in the optimum the moral hazard constraint binds, hence
\[
\alpha x+\tilde{\beta}y=\max\{\theta_{max}y+c,p_{1}y+2c\}
\]
Similar to the reputable agent, adverse selection binds as well, and
we have the following proposition. \begin{proposition}[Optimal whole-portfolio contract for non-reputable GP]\label{prop:whole-portfolio-non-reputable}

In the optimal whole-portfolio contract we have 
\[
y=p_{1}x
\]
 If $\rho\geq\rho^{*}=\frac{\lambda_{max}-\lambda_{min}}{\lambda_{min}(\frac{1}{p_{1}}-1)}$
then 
\[
y=\frac{2p_{1}c}{\alpha-p_{1}(p_{1}-\tilde{\beta})}
\]
 otherwise 
\[
y=\frac{p_{1}c}{\alpha-p_{1}(\theta_{max}-\tilde{\beta})}
\]
 Moreover $y=s^{FNO}(R+I)$ coincides with the payout to the GP in
the reputable case $s(R+I)=\frac{z}{p_{1}}$ from Proposition \ref{prop: Optimal-whole-portfolio}.

\end{proposition}

With the same reasoning as in Proposition \ref{prop: Optimal-whole-portfolio},
two cases are associated with the fact that the expected payout to
the GP ($\alpha x+\tilde{\beta}y)$ becomes $\theta_{max}y+c$ or
$p_{1}y+2c.$ In the non-reputable case, the total payout of both
projects is 
\[
2\alpha R+\tilde{\beta}(R+I)+(1-\alpha-\tilde{\beta})I.
\]
 The derivative with respect to $\rho$ of the total payout is
\[
\lambda_{min}(p_{1}R-I)<0.
\]
As we can see, on one hand the total payout of projects is decreasing
with respect to the correlation between projects. Intuitively this
is because the only scenario in which the investment decision is not
optimal is when both projects are bad and the chance of this scenario
is higher for higher correlation. On the other hand, the total payout
to the GP is decreasing with respect to correlation as again it gets
easier to motivate the GP to exert effort when the correlation goes
up. As a result, the total payout to the investor is ambiguous with
respect to correlation and depends on the relative magnitude of $p_{1}R-I$
and $c.$ Here $p_{1}R-I$ measures the inefficiency associated with
investing in the bad project with success chance $p_{1}$ (hence expected
return $p_{1}R)$ instead of investing in the safe asset (with return
$I$). When $c$ is large enough, since the expected payout to the
GP is proportional to $c,$the total payout to the GP decreases faster
compared to the loss of inefficiency which is proportional to $p_{1}R-I$.
As a result, by increasing $\rho$ the total payout to the LP increases
as well. If $c$ is small enough, the reverse phenomenon happens and
hence the total payout to the LP is decreasing with respect to $\rho$.
Finally, in the middle range of $c,$ total payout to the GP has an
interior optimal correlation. Hence we do not have a monotonic relationship
between the LP's payout and the correlation in the non-reputable case.
However, we can show the following.

\begin{proposition}[Reputable VS Non-reputable agent]\label{prop:reputable-VS-NR}

Suppose for parameters $(\lambda_{1},\lambda_{2},p_{1},p_{2},\rho),$the
investor prefers whole-portfolio contracting when writing contracts
with a generic (reputable) GP. This will also be the case when writing
a contract with a non-reputable GP. In particular, if $\rho\geq\rho^{*}=\frac{\lambda_{max}-\lambda_{min}}{\lambda_{min}(\frac{1}{p_{1}}-1)},$
then whole-portfolio contracting is dominant in the non-reputable
case as well.

\end{proposition}

This proposition comes from the fact that when dealing with non-reputable
agents, whole-portfolio contracting can help to improve the investment
strategy. Hence the space of parameters in which whole-portfolio contracting
is better for the LP is larger compared to the reputable case. As
mentioned in the introduction, this result has been observed empirically
in Huther et al (2020).

\section{Extension \label{sec:Extension}}

In this section, I consider various modifications of the model and
how it can affect the results. I make some predictions/observations
as well.

\subsection{No Asymmetric Information\label{sec:No-Asymmetric-Information}}

In the special case when $p_{1}=p_{2}=0,$ as we mentioned in section
\ref{sec:Model}, there is no profitable bad option for the GP to
invest in. Equivalently, there is no information asymmetry about the
quality of projects between the GP and the LP as there is no possible
profitable deviation. Under binary effort assumption, by equation
(\ref{eq:Profit-GP}), the whole surplus of every project goes to
the LP and hence the method of contracting is irrelevant in this setting.
In this case when adverse selection is absent, since effort is binary,
contract makes the GP indifferent between exerting effort and not
exerting effort and hence the LP can get the first-best outcome. In
order to analyze this important special case in depth, here I allow
for different levels of effort to see how it affects the contract.
Therefore for the purpose of this section, suppose with the variable
cost $c_{i}(\lambda_{i}),$ the chance of getting a good $i-$project
is $\lambda_{i}$. I assume $c_{i}(0)=c_{i}'(0)=c_{i}''(0)=0$ and
$c_{i}^{(3)}>0.$ I first consider the general case and then restrict
to the especial case $p_{1}=p_{2}=0$ which we are interested in.
As before, the chance of success for a $i-$project of type $B$ is
$p_{i}$ . The first-best effort satisfies 
\[
\max_{\lambda_{i}}\lambda_{i}R+(1-\lambda_{i})I-c_{i}(\lambda_{i})-I
\]
By FOC the optimal effort satisfies $c'_{i}(\lambda_{i}^{FB})=R-I$.
Not surprisingly it is independent of $p_{i}$ as there is no adverse
selection problem. I Assume $c'_{i}(1)>R-I$ to make sure that $\lambda_{i}^{FB}<1.$
Now suppose the contract $(s_{GP}(I),s_{GP}(R))$ is offered to the
GP. Similar to the binary case, the contract should satisfy 
\begin{equation}
s_{GP}(I)\geq p_{i}s_{GP}(R)\label{eq:general-adverse}
\end{equation}
to make sure that GP does not invest in the bad project. Once offered,
the agent chooses effort $\lambda_{i}$ which is the solution to the
problem 
\[
\max_{\lambda_{i}}\lambda_{i}s_{GP}(R)+(1-\lambda_{i})s_{GP}(I)-c_{i}(\lambda_{i})
\]
FOC implies 
\begin{equation}
s_{GP}(R)-s_{GP}(I)=c'_{i}(\lambda_{i})\label{eq:general-fb}
\end{equation}
in particular by decreasing $s_{GP}(I),$ the effort increases which
is in the favor of LP (both lower payment and higher effort). Hence
in the optimal $s_{GP}(I)=p_{i}s_{GP}(R)$ as contract should satisfies
equation (\ref{eq:general-adverse}). This gives 
\begin{equation}
c_{i}'(\lambda_{i}^{*})=(1-p_{i})s_{GP}(R)\label{eq:s(R)}
\end{equation}
in the optimum. Once we have this, $\lambda_{i}$ is determined by
solving
\[
\max_{\lambda_{i}}\lambda_{i}(R-I)-\frac{\lambda_{i}c'_{i}(\lambda_{i})+(1-\lambda_{i})p_{i}c'_{i}(\lambda_{i})}{1-p_{i}}
\]
which comes from the LP problem 
\[
\max_{\lambda_{i}}\lambda_{i}[R-s_{GP}(R)]+(1-\lambda_{i})[I-p_{i}s_{GP}(R)]
\]
and equation (\ref{eq:s(R)})$.$ Specializing to the case $p_{i}=0,$one
gets $c'_{i}(\lambda_{i}^{*})=s_{GP}(R).$ The equation to determine
$\lambda_{i}$ becomes 
\[
\max_{\lambda_{i}}\lambda_{i}(R-I)-\lambda_{i}c'_{i}(\lambda_{i})
\]
and FOC gives 
\[
R-I=c'_{i}(\lambda_{i}^{*})+\lambda_{i}c"_{i}(\lambda_{i}^{*}).
\]
When comparing the second-best effort with first best (equation (\ref{eq:general-fb})),
the extra term $\lambda_{i}c''_{i}(\lambda_{i})$ measures the moral
hazard issue and it reduces the effort by the agent. So in this more
flexible setting, first-best is not contactable even without adverse
selection. The expected payout from one project to the LP is 
\[
\lambda_{i}(R-c'_{i}(\lambda_{i}))+(1-\lambda_{i})I-I,
\]
which can be written as 
\begin{align*}
 & \lambda_{i}(R-I-c'_{i}(\lambda_{i}))\\
 & =\lambda_{i}^{2}c''_{i}(\lambda_{i}).
\end{align*}
The expected payout to the GP is $\lambda_{i}c'_{i}(\lambda_{i})$
and $\lambda_{i}c'_{i}(\lambda_{i})-c_{i}(\lambda_{i})>0$ is the
expected profit for the the GP. 

For the whole-portfolio contracting, I consider a simple whole-portfolio
contract which pays agent only when the outcome is $2R.$ In this
case, the GP problem is
\[
\max_{\lambda_{1},\lambda_{2}}\lambda_{1}\lambda_{2}s_{GP}(2R)-c_{1}(\lambda_{1})-c_{2}(\lambda_{2})
\]
 which gives 
\begin{align*}
 & c_{2}'(\lambda_{2})=\lambda_{1}s_{GP}(2R)\\
 & c_{1}'(\lambda_{1})=\lambda_{2}s_{GP}(2R)
\end{align*}
Once this, the LP problem is to determine $s_{GP}(2R)$ to maximize
the expected revenue which is 
\[
\max_{s_{GP}(2R)}\lambda_{1}\lambda_{2}(2R-s_{GP}(2R))+[\lambda_{1}(1-\lambda_{2})+\lambda_{2}(1-\lambda_{1})](R+I)+(1-\lambda_{1})(1-\lambda_{2})2I
\]
 where $\lambda_{1},\lambda_{2}$ satisfy GP's optimality equations.
We have the following proposition which gives answer for a wide class
of cost functions. 

\begin{proposition}[No asymmetric information case]\label{prop: general-cost}

If $c_{1}=a\lambda^{m}$ and $c_{2}=b\lambda^{m}$ for $m>2$ and
$a,b>0,$ then whole-portfolio contracting is better for the LP compared
to deal-by-deal contract.

\end{proposition}

The intuition behind the proposition is simple. When there is no asymmetry
of information, it is better that contract motivates effort as easily
as possible. Tying outcomes together can provide a bigger incentive
relative to contracting on projects in deal-by-deal basis. This proposition
shed light on the fact that in settings where the investor and agent
have same information about the quality of projects, the whole-portfolio
contracting is dominant. This includes hedge-funds, mutual-bonds or
other contracts on public equities. 

\subsection{Uninformed Investor\label{sec:Uninformed-Investor}}

In this part, I consider the case in which investor is not informed
about the correlation. Other than this, I assume the same setup as
in the main model. Since in the deal-by-deal contract correlation
has no effect on the outcome, I only focus on the whole-portfolio
contracting. Assume that the GP can privately and strategically choose
the correlation $\rho$ in the interval $[\rho_{1},\rho_{2}]$. While
the interval is common knowledge, the LP does not observe $\rho$
directly. The case of informed investor is a special case when $\rho_{1}=\rho_{2}=\rho$.
Since investment compatibility conditions are independent from $\rho,$
as in the informed case, optimal contract satisfies 
\begin{align*}
 & z\geq p_{1}y,p_{2}y,p_{1}p_{2}x\\
 & y\geq z,p_{1}x,p_{2}x\\
 & x\geq z,y
\end{align*}
where variables are as in equation (\ref{eq:payout-abr}). As before,
assume $p_{1}\geq p_{2}.$ The expected payout to the GP from choosing
$\rho$ is 
\[
\rho\lambda_{min}x+[\lambda_{1}-2\rho\lambda_{min}+\lambda_{2}]y+[1-\lambda_{1}-\lambda_{2}+\rho\lambda_{min}]z.
\]
Derivative with respect to $\rho$ of the expression above is 
\[
\lambda_{min}(x-2y+z).
\]
There are two possible scenarios for the GP to choose the correlation.
If $x>2y-z,$ the payout for two successful exits are relatively high
hence GP wants to maximize the chance of having two successful exits.
Therefore GP chooses the highest possible correlation i.e $\rho=\rho_{2}$.
In contrary, if $y$ is relatively high ($2y>x-z$), then it is more
profitable for GP to have only one successful exit. This event has
the highest chance when $\rho$ is smallest which is $\rho=\rho_{1}.$
Now suppose $(x,y,z)$ has the form of $(\frac{z}{p_{1}^{2}},\frac{z}{p_{1}},z)$
which is the same as in the optimal contract with informed investor
from Proposition \ref{prop: Optimal-whole-portfolio}. In this case,
the derivative with respect to $\rho$ of the payout to GP becomes
\[
\lambda_{min}z(\frac{1}{p_{1}^{2}}-\frac{2}{p_{1}}+1)=\lambda_{min}z(\frac{1}{p_{1}}-1)^{2}>0.
\]
Therefore as argued above, GP chooses the highest value of $\rho$
which is $\rho_{2}.$ As we saw in the informed problem, when $\rho$
goes up, $E[s_{LP}]$ goes up as well in the optimal informed contract.
Therefore LP optimally offers the contract $(\frac{z^{*}}{p_{1}^{2}},\frac{z^{*}}{p_{1}},z^{*})$
where $z^{*}=z^{*}(\rho_{2})$ is the management fee in the optimal
contract for the correlation $\rho_{2}$ from Proposition \ref{prop: Optimal-whole-portfolio}.
GP optimally chooses $\rho_{2}$ as well from discussion above. Therefore
if for $\rho=\rho_{2}$, LP prefers whole-portfolio contracting to
deal-by-deal, then the contract above is offered. Otherwise deal-by-deal
contract is offered. In summary we have 

\begin{proposition}[Uninformed Investor]\label{prop-uninformed LP}

Suppose GP can privately chooses $\rho$ in the interval $[\rho_{1},\rho_{2}].$Then 
\begin{itemize}
\item If $\rho_{2}<\rho^{**},$ a deal-by-deal contract is offered to GP.
\item If $\rho_{2}\geq\rho^{**},$a whole-portfolio contract associated
to $\rho=\rho_{2}$ is offered to GP and GP chooses $\rho=\rho_{2}$
optimally.
\end{itemize}
\end{proposition}

\subsection{Conditional Contract \label{sec:Conditional-Contract}}

In this part, I consider conditional contracting which means that
the payout of the contract is a function of the outcome of each project.
This definition contains both deal-by-deal contracting as well as
whole-portfolio contracting as special cases. When the payout of projects
are $2I$ or $2R,$ it corresponds uniquely to two bad or good projects
respectively. However, there are two possible ways to get the outcome
$R+I$. When the first project is type $G$ or when the second one
is and the other one is type $B$. Unlike the whole-portfolio and
similar to deal-by-deal, when the contract is fully conditional, total
payout to GP can be different in these two cases. Therefore take $y_{1}$
and $y_{2}$ as possible payouts to GP where $y_{1}$ is $s(R+I)$
when the first project is successful and $y_{2}$ is that of when
the second one is type $G$. Set $x=s(2R)$ and $z=s(2I)$ as in the
whole-portfolio contracting. Also recall that from equation (\ref{eq:s(I)=00003Ds(R)=00003D0}),
we have $s(I)=s(R)=0.$ Similar to whole-portfolio contracting, in
the optimal contract, payouts satisfy 
\begin{align*}
 & z\geq p_{1}p_{2}x,p_{1}y_{1},p_{2}y_{2}\\
 & y_{1}\geq z,p_{2}x\\
 & y_{2}\geq z,p_{1}x\\
 & x\geq z,y_{1},y_{2}
\end{align*}
When comparing these to analogues inequalities in the whole-portfolio
contracting, we see that there is efficiency gain. recall that in
the whole-portfolio contracting, $z$ should be bigger than both $p_{1}y$
and $p_{2}y$ since there is no difference between payouts to GP when
the return is $R+I$ and the first project is successful or the return
is $R+I$ and the second project is successful. As a result LP should
overcompensate GP to cover both cases and this causes some inefficiency
when compared to conditional contracting. The same phenomena happens
when comparing $y$ and $x$. $y$ should be bigger than both $p_{1}x$
and $p_{2}x$ while in the conditional contract there are two different
values $y_{1}$ and $y_{2}$ instead of single payout $y$. Assume
$y_{min}$ and $y_{max}$ are the corresponding payouts for when the
project with $\lambda_{min}$ or $\lambda_{max}$ succeed respectively.
To motivate GP to exert effort on both projects (moral hazard problem),
the contract should satisfy 
\begin{align*}
 & \rho\lambda_{min}x+(\lambda_{max}-\rho\lambda_{min})y_{max}+(1-\rho)\lambda_{min}y_{min}+(1-\lambda_{1}-\lambda_{2}+\rho\lambda_{min})z\\
 & \geq z+2c,\lambda_{1}y_{1}+(1-\lambda_{1})z+c,\lambda_{2}y_{2}+(1-\lambda_{2})z+c
\end{align*}
So LP problem is 
\begin{align*}
 & \min_{x,y,z}\alpha x+\beta_{1}y_{1}+\beta_{2}y_{2}+\gamma z\\
 & \alpha x+\beta_{1}y_{1}+\beta_{2}y_{2}+\gamma z\geq z+2c,\lambda_{i}y_{i}+(1-\lambda_{i})z+c\\
 & \ensuremath{x\geq y_{i}\geq z\geq p_{i}y_{i}\geq p_{1}p_{2}x}
\end{align*}
where $(\alpha,\beta_{1},\beta_{2},\gamma)=(\rho\lambda_{min},\lambda_{max}-\rho\lambda_{min},(1-\rho)\lambda_{min},1-\lambda_{1}-\lambda_{2}+\rho\lambda_{min})$.
With similar reasoning as in the whole-portfolio contracting, in the
optimum 
\begin{equation}
\alpha x+\beta_{1}y_{1}+\beta_{2}y_{2}+\gamma z=\max\{z+2c,\lambda_{i}y_{i}+(1-\lambda_{i})z+c\}\label{eq:GP-Payout-Conditional}
\end{equation}
Similar to Proposition \ref{prop: Optimal-whole-portfolio} we have
the following. 

\begin{proposition}[Optimal Conditional Contract]\label{prop: Optimal conditional}

In the optimal contract, 
\begin{equation}
z=p_{1}y_{1}=p_{2}y_{2}=p_{1}p_{2}x\label{eq:optimal-conditional}
\end{equation}
$z$ is determined such that equation (\ref{eq:GP-Payout-Conditional})
is satisfied.

\end{proposition}

As in the whole-portfolio contracting, the inequalities which are
dealing with adverse selection issue are binding as well. It is worth
mentioning that when $p_{1}=p_{2},$ we get $y_{1}=y_{2}$ even in
the case that $\lambda_{1}\not=\lambda_{2},$hence the contract reduces
to a whole-portfolio contract. This is because the variation in $\lambda$s
affect the moral hazard problem and has no bearing on the asymmetric
information issue which rises after exerting effort on projects. In
fact when we compare whole-portfolio contracting with the conditional
contracting for the loss of efficiency for LP, we have (recall that$p_{2}<p_{1})$

\begin{proposition}[Efficiency loss for whole-portfolio contracting]\label{prop:effciency-loss}

The profit made by GP in the conditional contracting decreases as
$p_{2}(\leq p_{1})$ increases. When $p_{2}=p_{1},$the profit equals
the profit made in the whole-portfolio contracting.

\end{proposition}

Here is the intuition behind statement above. . As we mentioned in
the discussion after proposition \ref{prop:AI-VS-MH-Reputable}, $p_{1}-p_{2}$
measures the difference between adverse selection that problems are
subject to. When this measure is zero, handling of the problem for
one project, efficiently takes care of the other project as well hence
whole-portfolio contracting becomes the best conditional contract.
As this measure grows, whole-portfolio contracting becomes less and
less efficient and GP can extract more rent on the projects.

\subsection{Bargaining Power\label{sec:GP-bargains}}

In this part, I consider what happens when GP has the bargaining power
for writing the contract. For the moment I assume that the offered
contract has to be incentive compatible to induce the GP to choose
the optimal strategy (justify it at the end). For reputable GP, the
problem becomes 
\begin{align*}
 & \max_{x,y,z}\alpha x+\beta y+\gamma z\\
 & \ensuremath{x\geq y\geq z\geq p_{1}y\geq p_{1}^{2}x};\;\;\;\alpha x+\beta y+\gamma z\geq z+2c,\lambda_{max}y+(1-\lambda_{max})z+c\\
 & E[s_{LP}]\geq2I
\end{align*}
The last inequality is the participation constraint for the LP. Here
notation are the same as in subsection \ref{subsec:Whole-Portfolio-Contract}.
Not surprisingly, in the optimum, the equality $E[s_{LP}]=2I$ happens
otherwise the contract can be altered in the GP's favor without violating
participation constraint by LP (for example by increasing $x$). Hence
unlike investor problem we do not have a unique contract and the contract
only needs to satisfy the incentive compatibility equations. Also,
as we saw in subsection \ref{subsec:Whole-Portfolio-Contract} , $E[s_{LP}]+E[s_{GP}]$
is independent of $\rho.$ Therefore, investor breaks even in the
optimum and the contract is not unique. To justify the imposing of
the incentive compatibility, note that in any feasible contract $E[s_{LP}]\geq2I$
is required. Also the optimal investment strategy guarantees the maximum
possible payout of the projects. Hence imposing them does not reduce
the profit by GP. \\
This result is not surprising as when GP has the market power, since
she is the party who takes the action and also observes the quality
of the project, she can extract all the rent from projects. So in
the presence of the market power by GP, the method of contracting
or correlation does not play a role.

When GP is not reputable hence FNO is imposed, take $E[s_{GP}]=\alpha x+\tilde{\beta}y$
as in section \ref{sec:Non-Reputable-GP} and then the GP problem
becomes 
\begin{align*}
 & \max_{x,y}\alpha x+\tilde{\beta}y\\
 & x\geq y\geq p_{1}x;\;\;\;\alpha x+\tilde{\beta}y\geq\max\{\theta_{max}y+c,p_{1}y+2c\}\\
 & E[s_{LP}]\geq2I
\end{align*}

Similar to the case of reputable GP, investor breaks even and he is
not concerned about the correlation. In addition, if GP can choose
the correlation as well, she maximizes the surplus of the project.
So as we saw in the section \ref{sec:Non-Reputable-GP}, total payout
is decreasing in $\rho$ so she chooses $\rho=0$ in this case. Since
even for $\rho=1,$ total payout of projects is more than deal-by-deal
contract, when GP has market power she chooses whole-portfolio contracting.
This stems from the fact that the whole-portfolio contracting persuade
GP to have a better investment strategy. Hence in summary we get 

\begin{proposition}

When GP has the bargaining power, for reputable GP there is no difference
between deal-by-deal and whole-portfolio contracting. For non-reputable
GP whole-portfolio contracting is preferred for all values of $\rho$.
In addition the contract is not unique and only needs to satisfy the
IC by GP as well as PC by LP. 

\end{proposition}

\section{Conclusion\label{sec:Conclusion}}

In this paper, I proposed a framework to study the scheme of payment
in LPA. Unlike usual contracts which only determine the amount of
payment for a given return, since GP and LP write a contract on a
portfolio, the method of payment is also of vital importance. I compared
the main two methods of payments which are the deal-by-deal and the
whole-portfolio. Within my setting, I showed that the whole-portfolio
contracting is more prevalent when the correlation of investment companies
is high or when the reputation of GP is low. Previously documented
findings support these result.\\
\\
In addition, I make some predictions which can guide future studies.
For example the informativeness of the investor can also affect the
method of contracting and hence the portfolio as well. More informed
investors tend to have more deal-by-deal contract and a diverse portfolio
while less informed ones have a narrower range of investment and more
whole-portfolio contracting. Also when underlying assets are public,
whole-portfolio contracting is the typical method of payment which
is used.

\section{Appendix A: Proofs\label{sec:Appendix-A:-Proofs}}

\subsection{Proof of Proposition \ref{prop: Optimal conditional}}

Let's recall the equations which LP should consider to design the
security. We have 
\begin{align*}
 & z\geq p_{1}p_{2}x,p_{1}y_{1},p_{2}y_{2}\\
 & y_{1}\geq z,p_{2}x\\
 & y_{2}\geq z,p_{1}x\\
 & x\geq z,y_{1},y_{2}
\end{align*}
 where as defined before, $(x,z)=(s(2R),s(2I))$ and $y_{i}$ is the
payout to GP when project $i$ is successful and project $3-i$ is
not. As we saw, LP problem is 
\begin{align*}
 & \min_{x,y,z}\alpha x+\beta_{1}y_{1}+\beta_{2}y_{2}+\gamma z\\
 & \alpha x+\beta_{1}y_{1}+\beta_{2}y_{2}+\gamma z\geq z+2c,\lambda_{i}y_{i}+(1-\lambda_{i})z+c\\
 & \ensuremath{x\geq y_{i}\geq z\geq p_{i}y_{i}\geq p_{1}p_{2}x}
\end{align*}
where $(\alpha,\beta_{1},\beta_{2},\gamma)=(\lambda_{1}\lambda_{2},\lambda_{1}-\rho\lambda_{2},(1-\rho)\lambda_{2},1-\lambda_{1}-\lambda_{2}+\rho\lambda_{2})$.
In the discussion before the proposition, we showed that
\begin{equation}
\alpha x+\beta_{1}y_{1}+\beta_{2}y_{2}+\gamma z=\max\{z+2c,\lambda_{i}y_{i}+(1-\lambda_{i})z+c\}\label{eq:GP-payout-conditional}
\end{equation}

Now we claim that, in the optimum either $z=p_{1}y_{1}$ or $z=p_{2}y_{2}$.
This happens since if $z$ can be reduced, the coefficient of $z$
on the LHS which is $\gamma=1-\lambda_{1}-\lambda_{2}+\rho\lambda_{2}$
is less than (or equal to if $\rho=1)$ the minimum of coefficients
on the RHS of equation (\ref{eq:GP-payout-conditional}) which is
$\min(1-\lambda_{1},1-\lambda_{2}).$ Having this, assume that $z=p_{1}y_{1}\geq p_{2}y_{2}.$
Consider several cases
\begin{itemize}
\item First assume that the maximum on the RHS of equation (\ref{eq:GP-payout-conditional})
happens at $z+2c.$ Then we get
\[
\alpha x=z(1-\frac{\beta_{1}}{p_{1}}-\gamma)-\beta_{2}y_{2}+2c.
\]
 After multiplying by $p_{1}p_{2}$ and using the facts that $y_{2}\leq\frac{z}{p_{2}},x\leq\frac{z}{p_{1}p_{2}}$
we have 
\begin{align*}
 & z[\alpha(1-p_{1}p_{2})+p_{2}\beta_{1}(1-p_{1})+p_{1}\beta_{2}(1-p_{2})]\\
 & \geq2p_{1}p_{2}c
\end{align*}
 with strict inequality if either $y_{2}<\frac{z}{p_{2}}$ or $x<\frac{z}{p_{1}p_{2}}$.
This gives a lower bound for $z$. In the optimum the inequality becomes
equality (as LP looks for the minimum payment to GP) hence we get
$z=p_{1}y_{1}=p_{2}y_{2}=p_{1}p_{2}x.$
\item Now assume that maximum happens at $\lambda_{1}y_{1}+(1-\lambda_{1})z+c.$
Then similar to the previous case, we can write 
\[
\alpha x=\lambda_{1}\frac{z}{p_{1}}+(1-\lambda_{1})z-\beta_{1}\frac{z}{p_{1}}-\beta_{2}y_{2}-\gamma z+c
\]
 Again inequalities as in the previous case gives us a lower bound
for $z$ which is binding in the optimum hence we get the same relationship
between variables.
\item Finally assume that maximum happens at $\lambda_{2}y_{2}+(1-\lambda_{2})z+c$
and assume this is strictly bigger than other two terms ($z+2c$ and
$\lambda_{1}y_{1}+(1-\lambda_{1})z+c).$ If $z=p_{1}y_{1}=p_{2}y_{2},$
then the same reasoning as in the previous case works. So assume that
$z=p_{1}y_{1}>p_{2}y_{2}\geq p_{1}p_{2}x.$ Now reduce $z$ by $\epsilon$
and $y_{1}$ by $\frac{\epsilon}{p_{1}}$ and change $x$ accordingly
to $x'$ such that 
\[
\alpha x'+\beta_{1}(y_{1}-\frac{\epsilon}{p_{1}})+\beta_{2}y_{2}+\gamma(z-\epsilon)=\lambda_{2}y_{2}+(1-\lambda_{2})(z-\epsilon)+c
\]
 If this change is permissible, then the new contract $(x',y_{1}-\frac{\epsilon}{p_{1}},y_{2},z-\epsilon)$
is strictly better for LP. So it would not be permissible. If $x'>x,$this
means that we should have had $p_{1}x=y_{2}$ so that an increase
in $x$ is not permissible. However in this case, the original equation
for the expected payout to GP, becomes 
\[
\alpha(x-y_{2})=\alpha x(1-p_{1})=z(1-\frac{\beta_{1}}{p_{1}}-\gamma-\lambda_{2})+c
\]
 so 
\[
\alpha\frac{z}{p_{1}p_{2}}(1-p_{1})\geq z(1-\frac{\beta}{p_{1}}-\gamma-\lambda_{2})+c
\]
 and again in the optimum this should be equality hence $z=p_{1}p_{2}x$
which proves the claim. If $x'<x,$ then either $x=y_{2}$ or $x=y_{1}.$
But these contracts can not be optimal since reducing all the payouts
$z,y_{1}$ and $y_{2}$ and increasing $x$ is allowed here which
improves payout to LP. 
\end{itemize}

\subsection{Proof of Proposition\ref{prop: Optimal-whole-portfolio} and \ref{prop: WP-VS-DBD-Reputable}}

Let's recall the LP problem. 
\begin{align*}
 & \min_{x,y,z}\alpha x+\beta y+\gamma z\\
 & \alpha x+\beta y+\gamma z\geq z+2c,\lambda_{max}y+(1-\lambda_{max})z+c\\
 & \ensuremath{x\geq y\geq z\geq p_{1}y\geq p_{1}^{2}x}
\end{align*}
where $(\alpha,\beta,\gamma)=(\rho\lambda_{min},\lambda_{1}-2\rho\lambda_{min}+\lambda_{2},1-\lambda_{1}-\lambda_{2}+\rho\lambda_{min})$.
As in the case of Proposition \ref{prop: Optimal conditional}, we
make two preliminary observations. 
\begin{enumerate}
\item In the optimum, we have $\alpha x+\beta y+\gamma z=\max\{z+2c,\lambda_{max}y+(1-\lambda_{max})z+c\}.$
Suppose not. Then the transformation $z\rightarrow z-\epsilon$ for
small enough $\epsilon,$ should violate the conditions of the LP
problem otherwise reducing $z$ improves the payout to LP. This means
$z=p_{1}y$. But then $y\rightarrow y-\epsilon$ should violate the
conditions so similarly we get $y=p_{1}x$ as well. But then $x\rightarrow x-\epsilon$
is legitimate because by $z=p_{1}y=p_{1}^{2}x$ we have $x>y>z.$ 
\item In the optimum, $z=p_{1}y.$ This comes from the fact that in the
LP problem $\gamma=1-\lambda_{1}-\lambda_{2}+\rho\lambda_{min}\leq1-\lambda_{1},1-\lambda_{2}$.
So if the move $z\rightarrow z-\epsilon$ is allowed, the subtraction
on the LHS which is $\gamma\epsilon$ is less than the subtraction
on the RHS which is either $\epsilon$ or $(1-\lambda_{max})\epsilon.$
\end{enumerate}
By point 1 above, the expected payout to GP will be $\max\{z+2c,\lambda_{max}y+(1-\lambda_{max})z+c\}.$
This function is increasing in $z$, equals to $z+2c$ for $z\leq\frac{p_{1}c}{\lambda_{max}(1-p_{1})}$
(since $z=p_{1}y)$ and it is $\lambda_{max}y+(1-\lambda_{max})z+c$
otherwise. Because the payout is increasing in $z,$ GP searches for
a feasible contract with least amount of $z.$ First we see when LP
is able to write a contract with $z\leq\frac{p_{1}c}{\lambda_{max}(1-p_{1})}.$
Here is how the contract is designed. As we saw, in this range of
$z$, 
\[
\alpha x+\beta y+\gamma z=\max\{z+2c,\lambda_{max}y+(1-\lambda_{1})z+c\}=z+2c
\]
 therefore 
\begin{equation}
x=\frac{z(1-\gamma-\frac{\beta}{p_{1}})+2c}{\alpha}\label{eq:x-small z}
\end{equation}
hence GP should solve 
\begin{align*}
 & \min z\\
 & x\geq\frac{z}{p_{1}}\geq p_{1}x\;\;\;z\leq\frac{p_{1}c}{\lambda_{max}(1-p_{1})}
\end{align*}
given $x$ as in the equation (\ref{eq:x-small z}). $p_{1}x\geq z$
implies 
\[
2p_{1}c\geq(1-p_{1})(\alpha+\beta)z
\]
 which implies 
\[
z\leq\frac{2p_{1}c}{(\alpha+\beta)(1-p_{1})}
\]
 On the other hand $z\geq p_{1}^{2}x,$ implies 
\[
[\beta\frac{1-p_{1}}{p_{1}}+\alpha\frac{1-p_{1}^{2}}{p_{1}^{2}}]z\geq2c
\]
 which gives 
\begin{equation}
z\geq\frac{2c}{\beta\frac{1-p_{1}}{p_{1}}+\alpha\frac{1-p_{1}^{2}}{p_{1}^{2}}}\label{eq:z-small z}
\end{equation}
Since 
\[
\frac{2c}{\beta\frac{1-p_{1}}{p_{1}}+\alpha\frac{1-p_{1}^{2}}{p_{1}^{2}}}<\frac{2p_{1}c}{(\alpha+\beta)(1-p_{1})}
\]
 one needs to compare the lower bound on the $z$ from equation (\ref{eq:z-small z})
with the initial condition $z\leq\frac{p_{1}c}{\lambda_{max}(1-p_{1})}$.
This gives us that there is answer in this region if and only if
\[
[\beta\frac{1-p_{1}}{p_{1}}+\alpha\frac{1-p_{1}^{2}}{p_{1}^{2}}]\geq\frac{2\lambda_{max}(1-p_{1})}{p_{1}}
\]
 which simplifies to 
\[
\beta+\alpha\frac{1+p_{1}}{p_{1}}\geq2\lambda_{max}
\]
 Substituting $\alpha$ and $\beta$ gives us that this happens if
and only if 
\[
\rho\geq\frac{\lambda_{max}-\lambda_{min}}{\lambda_{min}(\frac{1}{p_{1}}-1)}
\]
In this case the contract is written with $z$ as given by equality
in the equation (\ref{eq:z-small z}) which is the minimal $z$ hence
$z=p_{1}y=p_{1}^{2}x$. Now consider the case $\rho<\frac{\lambda_{max}-\lambda_{min}}{\lambda_{min}(\frac{1}{p_{1}}-1)}$.
The expected payout to GP then has the form $\lambda_{max}y+(1-\lambda_{max})z+c$
for some $z>\frac{p_{1}c}{\lambda_{max}(1-p_{1})}.$ In this case
we get 
\[
\alpha x+\beta y+\gamma z=\lambda_{max}y+(1-\lambda_{max})z+c=(\frac{\lambda_{max}}{p_{1}}+1-\lambda_{max})z+c
\]
 hence
\[
x=\frac{(\frac{\lambda_{max}}{p_{1}}+1-\lambda_{max}-\frac{\beta}{p_{1}}-\gamma)z+c}{\alpha}
\]
 Similar to the previous case, the optimization becomes 
\begin{align*}
 & \min z\\
 & x\geq\frac{z}{p_{1}}\geq p_{1}x\;\;\;z\geq\frac{p_{1}c}{\lambda_{max}(1-p_{1})}
\end{align*}
 $p_{1}x\geq z$ implies 
\[
[\lambda_{max}+p_{1}-p_{1}\lambda_{max}-\beta-p_{1}\gamma]z+p_{1}c\geq\alpha z
\]
 which is 
\[
p_{1}c\geq(\alpha+\beta-\lambda_{max})(1-p_{1})z
\]
 We have $\alpha+\beta-\lambda_{max}=\lambda_{min}(1-\rho)$ 
\[
z\leq\frac{p_{1}c}{(1-p_{1})\lambda_{min}(1-\rho)}
\]
Finally $z\geq p_{1}^{2}x$ implies 
\[
(\frac{\lambda_{max}}{p_{1}}+1-\lambda_{max}-\frac{\beta}{p_{1}}-\gamma)z+c\leq\frac{\alpha z}{p_{1}^{2}}
\]
 which implies 
\[
c\leq[\frac{\alpha}{p_{1}^{2}}-(\frac{\lambda_{max}}{p_{1}}+1-\lambda_{max}-\frac{\beta}{p_{1}}-\gamma)]z
\]
 We can write this as 
\[
c\leq[\frac{\alpha(1-p_{1}^{2})}{p_{1}^{2}}+\beta\frac{1-p_{1}}{p_{1}}-\lambda_{max}\frac{1-p_{1}}{p_{1}}]z
\]
therefore again the contract has the following form
\begin{align*}
 & z=\frac{c}{\frac{\alpha(1-p_{1}^{2})}{p_{1}^{2}}+\beta\frac{1-p_{1}}{p_{1}}-\lambda_{max}\frac{1-p_{1}}{p_{1}}}\\
 & y=\frac{z}{p_{1}}\\
 & x=\frac{z}{p_{1}^{2}}
\end{align*}
Finally to check the upper bound $\frac{p_{1}c}{(1-p_{1})\lambda_{min}(1-\rho)}$
for $z,$ one needs to verify 
\[
\frac{\alpha(1-p_{1}^{2})}{p_{1}}+\beta(1-p_{1})-\lambda_{max}(1-p_{1})\geq\lambda_{min}(1-p_{1})(1-\rho)
\]
Canceling $1-p_{1},$we get 
\[
\frac{\alpha(1+p_{1})}{p_{1}}+\beta-\lambda_{max}\geq\lambda_{min}(1-\rho)
\]
LHS equals to 
\[
\lambda_{min}(\rho(\frac{1}{p_{1}}-1)+1)\geq\lambda_{min}(1-\rho)
\]
which is obvious. 

The only part from Proposition \ref{prop: WP-VS-DBD-Reputable}, which
needs proof is that when $\rho\geq\frac{\lambda_{max}-\lambda_{min}}{\lambda_{min}(\frac{1}{p_{1}}-1)},$
then LP makes more profit by whole-portfolio contracting. This is
because, as we saw in the proof above, for these values of $\rho,$the
expected profit by GP with whole-portfolio contract is $z+2c$ for
some $z\leq\frac{p_{1}c}{\lambda_{max}(1-p_{1})}$ which is less than
$\Pi_{GP}=2c+\sum_{i=1}^{2}\frac{p_{i}c}{\lambda_{i}(1-p_{i})}$ that
GP makes under the deal-by-deal contract. 

\subsection{Proof of Proposition \ref{prop:Comparative-statics}}

As we saw in the Proof of Proposition \ref{prop: Optimal-whole-portfolio},
when $\rho>\rho^{*},$ the expected payout to GP is $z+2c.$ The expression
for $z$ in this region is 
\[
z=\frac{2c}{\beta\frac{1-p_{1}}{p_{1}}+\alpha\frac{1-p_{1}^{2}}{p_{1}^{2}}}
\]
 So it is only needed to look at how denominator changes when parameters
change. Denominator is equal to 
\begin{align*}
 & \beta\frac{1-p_{1}}{p_{1}}+\alpha\frac{1-p_{1}^{2}}{p_{1}^{2}}\\
 & =\frac{1-p_{1}}{p_{1}}[\alpha\frac{1+p_{1}}{p_{1}}+\beta]
\end{align*}
Derivative with respect to $p_{1}$ becomes 
\[
-\frac{1}{p_{1}^{2}}[\alpha\frac{1+p_{1}}{p_{1}}+\beta]+\frac{1-p_{1}}{p_{1}}\times\alpha\frac{-1}{p_{1}^{2}}<0
\]
With respect to $\lambda_{min}$ and $\lambda_{max}$, derivatives
are $\frac{1-p_{1}}{p_{1}}[\rho\frac{1+p_{1}}{p_{1}}+1-2\rho]$ and
$\frac{1-p_{1}}{p_{1}}$ and both are positive. This proves the proposition
in the region $\rho>\rho^{*}.$ In the region $\rho<\rho^{*},$the
expected payout to GP is $\lambda_{max}y+(1-\lambda_{max})z=\frac{\lambda_{max}(1-p_{1})+p_{1}}{p_{1}}z.$
The expression for $z$ in this area is 
\[
z=\frac{c}{\frac{\alpha(1-p_{1}^{2})}{p_{1}^{2}}+\beta\frac{1-p_{1}}{p_{1}}-\lambda_{max}\frac{1-p_{1}}{p_{1}}}
\]
 with respect to $\lambda_{min},$the derivative of the denominator
is same as above. With respect to $\lambda_{max}$, the derivative
of $z$ is zero. However since the expected payout in this regime
is $\frac{\lambda_{max}(1-p_{1})+p_{1}}{p_{1}}z,$it is increasing.
Finally with respect to $p_{1},$the denominator for $\frac{1-p_{1}}{p_{1}}z$,
is $\frac{\alpha(1+p_{1})}{p_{1}}+\beta-\lambda_{max}.$ Derivative
with respect to $p_{1}$ of this term is $-\frac{\alpha}{p_{1}^{2}}<0$
hence denominator is decreasing and the whole term is increasing.
This finishes the argument. 

\subsection{Proof of Proposition \ref{prop:whole-portfolio-non-reputable}}

As we mentioned in the discussion proceeding the Proposition \ref{prop:whole-portfolio-non-reputable},
in the optimal contract we have 
\[
\alpha x+\tilde{\beta}y=\max\{\theta_{max}y+c,p_{1}y+2c\}
\]
By this, LP problem can be written as
\begin{align*}
 & \min_{x,y}\alpha x+\tilde{\beta}y\\
 & x\geq y\geq p_{1}x;\;\;\;\alpha x+\tilde{\beta}y=\max\{\theta_{max}y+c,p_{1}y+2c\}
\end{align*}
Similar to the proof of Proposition\ref{prop: Optimal-whole-portfolio},
we consider two possible cases for $y$.
\begin{itemize}
\item If $y\leq\frac{c}{\theta_{max}-p_{1}}=\frac{c}{\lambda_{max}(1-p_{1})}$,
then we $p_{1}y+2c=\max\{\theta_{max}y+c,p_{1}y+2c\},$hence one gets
\[
x=f(y)=\frac{(p_{1}-\tilde{\beta})y+2c}{\alpha}
\]
 by $p_{1}y+2c=\alpha x+\tilde{\beta}y.$ Therefore LP problem becomes
\begin{align*}
 & \min y\\
 & \min\{f(y),\frac{c}{\theta_{m}-p_{1}}\}\geq y\geq p_{1}f(y)
\end{align*}
 $y\geq p_{1}f(y)$ implies 
\begin{equation}
y\geq\frac{2p_{1}c}{\alpha-p_{1}(p_{1}-\tilde{\beta})}\label{eq:trivial-y-NR}
\end{equation}
 $x\geq y$ implies 
\[
(\alpha-(p_{1}-\tilde{\beta}))y\leq2c
\]
 if $\alpha-(p-\tilde{\beta})>0,$ then 
\[
y\leq\frac{2c}{(\alpha-(p_{1}-\tilde{\beta}))}
\]
otherwise always (\ref{eq:trivial-y-NR}) is satisfied. Note that
if $\alpha-(p-\tilde{\beta})>0,$ then
\[
\frac{2p_{1}c}{\alpha-p_{1}(p_{1}-\tilde{\beta})}\leq\frac{2c}{(\alpha-(p_{1}-\tilde{\beta}))}
\]
So there is answer satisfying $y\leq\frac{c}{\theta_{max}-p_{1}}$
if
\begin{equation}
\frac{2p_{1}c}{\alpha-p_{1}(p_{1}-\tilde{\beta})}\leq\frac{c}{\theta_{max}-p_{1}}\label{eq:small-y-NR}
\end{equation}
In which case $y=\frac{2p_{1}c}{\alpha-p_{1}(p_{1}-\tilde{\beta})}$
and $x=f(y)=\frac{y}{p_{1}}.$ Equation (\ref{eq:small-y-NR}) holds
if 
\[
\alpha-p_{1}^{2}+p_{1}\tilde{\beta}\geq2\lambda_{max}p_{1}(1-p_{1})
\]
 which, after canceling and factoring $(1-p_{1}),$becomes 
\[
\lambda_{min}[(1-p_{1})\rho+p_{1}]\geq\lambda_{max}p_{1}
\]
 which is equivalent to 
\[
\rho\geq\frac{\lambda_{max}-\lambda_{min}}{\lambda_{min}(\frac{1}{p_{1}}-1)}
\]
\item If $\rho<\frac{\lambda_{max}-\lambda_{min}}{\lambda_{min}(\frac{1}{p_{1}}-1)}$,
then the contract with $y\leq\frac{c}{\lambda_{max}(1-p_{1})}$ is
not possible. In this case, 
\[
f(y)=x=\frac{(\theta_{max}-\tilde{\beta})y+c}{\alpha}
\]
 So LP should solve 
\begin{align*}
 & \min y\\
 & f(y)\geq y\geq\max\{p_{1}f(y),\frac{c}{\lambda_{max}(1-p_{1})}\}
\end{align*}
 $y\geq p_{1}f(y)$ implies 
\[
y\geq\frac{p_{1}c}{\alpha-p_{1}(\theta_{max}-\tilde{\beta})}
\]
 $x\geq y$ implies 
\[
(\alpha-(\theta_{max}-\tilde{\beta}))y\leq c
\]
 if $\alpha-(\theta_{max}-\tilde{\beta})>0$ this means 
\[
y\leq\frac{c}{(\alpha-(\theta-\tilde{\beta}))}
\]
 Otherwise always it is satisfied. In any case, in this region, the
possible minimum for $y$ is $\frac{p_{1}c}{\alpha-p_{1}(\theta_{max}-\tilde{\beta})},$
in which case $x=\frac{y}{p_{1}}.$
\end{itemize}
Finally to show the equality $s^{FNO}(R+I)=s(R+I),$ note that $\tilde{\beta}=\beta+p_{1}(1-\alpha-\beta).$
Hence in the region $\rho<\frac{\lambda_{max}-\lambda_{min}}{\lambda_{min}(\frac{1}{p_{1}}-1)},$
it is enough to show 
\begin{align*}
 & \alpha(1-p_{1}^{2})+\beta p_{1}(1-p_{1})-\lambda_{max}p_{1}(1-p_{1})=\\
 & \alpha-p_{1}(\lambda_{max}+p_{1}(1-\lambda_{max}))+p_{1}(\beta+p_{1}(1-\alpha-\beta))
\end{align*}
 which is correct after simplification. In the region $\rho\geq\frac{\lambda_{max}-\lambda_{min}}{\lambda_{min}(\frac{1}{p_{1}}-1)}$
similar algebra works.

\subsection{Proof of Proposition \ref{prop:reputable-VS-NR}}

We use the fact that $s^{FNO}(R+I)=s(R+I)$ which was proved in Proposition
\ref{prop:whole-portfolio-non-reputable}. I claim, expected payout
to both types of GPs are the same under whole-portfolio contracting.
For $\rho<\rho^{*}$ both are equal to $p_{1}y+2c=z+2c.$ In the region
$\rho>\rho^{*},$ for non-reputable GP, expected payout is equal to
\begin{align*}
 & \theta_{max}y+c=\\
 & \lambda_{max}y+(1-\lambda_{max})p_{1}y+c
\end{align*}
 By Proposition \ref{prop: Optimal-whole-portfolio} is the same as
$\lambda_{max}y+(1-\lambda_{max})z+c$ which is the expected payout
to reputable GP by the same proposition. Also as we saw by equation
(\ref{eq:profit-GP-FNO}) and discussion after it, the expected payouts
to both types of GP under deal-by-deal are the same as well. Now in
the reputable case, both types of contracting induce the optimal investment
strategy of exerting effort and invest only in good project. However
for non-reputable GP, whole-portfolio contracting improves total payout
of the projects compared to deal-by-deal as we saw in Subsection \ref{subsec:Whole-Portfolio-Contracting-With}.
Therefore it can only increase the profit of the LP compared to the
reputable case since the expected payout to GP is the same for both
types of GPs. This completes the argument. 

\subsection{Proof of Proposition \ref{prop: general-cost}}

With given cost functions, in the deal-by-deal, the effort $\lambda_{1}$
is determined by 
\[
R-I=am\lambda_{1}^{m-1}+am(m-1)\lambda_{1}^{m-1}=am^{2}\lambda_{1}^{m-1}
\]
 so $\lambda_{1}=\sqrt[m-1]{\frac{R-I}{am^{2}}}$ and the profit from
first project is $\lambda_{1}^{2}c"(\lambda_{1})=a\lambda_{1}^{2}m(m-1)\lambda_{1}^{m-2}.$
For the second project we need only to change $a$ to $b$. In the
whole-portfolio case, using equations 
\begin{align*}
 & bm\lambda_{2}^{m-1}=c_{2}'(\lambda_{2})=\lambda_{1}s_{GP}(2R)\\
 & am\lambda_{1}^{m-1}=c_{1}'(\lambda_{1})=\lambda_{2}s_{GP}(2R)
\end{align*}
we have $am\lambda_{1}^{m}=bm\lambda_{2}^{m}$ which gives $\lambda_{2}=C\lambda_{1}$
where $C=\sqrt[m]{\frac{a}{b}}.$ Then second equation above gives
\[
\frac{a}{C}m\lambda_{1}^{m-2}=s_{GP}(2R)
\]
 Therefore LP problem can be written as ($\lambda_{1}=\lambda)$ 
\[
\max_{\lambda}C\lambda^{2}[2R-s_{GP}(2R)]+[\lambda(1-C\lambda)+C\lambda(1-\lambda)](R+I)+(1-\lambda)(1-C\lambda)2I-2I
\]
 From above $C\lambda^{2}s_{GP}(2R)=am\lambda^{m}$ and the profit
is 
\begin{align*}
 & 2C\lambda^{2}R-am\lambda^{m}+\lambda R-C\lambda^{2}R+\lambda I-C\lambda^{2}I\\
 & +C\lambda R-C\lambda^{2}R+C\lambda I-C\lambda^{2}I-2\lambda I-2C\lambda I+2C\lambda^{2}I
\end{align*}
 which simplifies to 
\[
(\lambda+C\lambda)(R-I)-am\lambda^{m}
\]
 FOC gives 
\[
(1+C)(R-I)=am^{2}\lambda^{m-1}
\]
 so 
\[
\lambda_{1}=\lambda=\sqrt[m-1]{\frac{(1+C)(R-I)}{am^{2}}}
\]
 Total profit by LP in this case can be written as 
\[
\lambda(1+C)(R-I)[1-\frac{1}{m}]
\]
 So in order to show that LP makes more money with whole-portfolio
compared to deal-by-deal, we have to show 
\[
\sqrt[m-1]{\frac{(1+C)}{a}}(1+C)>\sqrt[m-1]{\frac{1}{a}}+\sqrt[m-1]{\frac{1}{b}}
\]
since $b=\frac{a}{C^{m}},$ this simplifies to show 
\[
\sqrt[m-1]{(1+C)}(1+C)>1+\sqrt[m-1]{C^{m}}
\]
which is equivalent to (set $C=d^{m-1})$
\[
(1+d^{m-1})^{m}>(1+d^{m})^{m-1}
\]
Since the problem is symmetric with respect to $a$ and $b$ we can
assume $a\leq b$ hence $d\leq1,$which makes the inequality above
trivial.

\section{Appendix B : Robustness\label{sec:Appendix-B-:Robustness}}

\subsection{Increasing Assumption\label{subsec:Increasing-Assumption}}

The security defined in Proposition \ref{prop: Optimal-whole-portfolio}
is increasing on the set of possible payouts with positive probability
(i.e on equilibrium path). However if the payout of good project is
not $R$ for sure or GP makes a wrong decision, it is possible to
get a payout of size $R$. To show robustness of our main result on
the relation between correlation and security design, we have the
following proposition. 

\begin{proposition}

If the condition $s(R)\geq s(2I)$ is imposed to the security $s_{GP}=s$,
the result of Proposition \ref{prop: WP-VS-DBD-Reputable} remains
unchanged if $\frac{1}{4}\geq p_{1}\geq p_{2}.$

\end{proposition}

\begin{proof}

Take $(z,w,x,y)=(s(2I),s(R),s(R+I),s(2R)).$ Similar to the case of
security with no restriction, in the optimum $s(I)=0.$ Also since
$R$ is not outcome of optimal investment strategy, it should be the
lowest possible value such that the security remains increasing hence
$s(R)=s(2I)=z$. With the same reason as discussed in subsection \ref{subsec:Whole-Portfolio-Contract},
to motivate for optimal investment strategy, contract should satisfy
\begin{align*}
 & z\geq p_{1}p_{2}x+[p_{1}(1-p_{2})+p_{2}(1-p_{1})]z,p_{1}y\\
 & y\geq p_{1}x+(1-p_{1})z\\
 & x\geq z,y
\end{align*}
Also to motivate effort, 
\begin{align*}
 & \rho\lambda_{min}x+[\lambda_{1}-2\rho\lambda_{min}+\lambda_{2}]y+[1-\lambda_{1}-\lambda_{2}+\rho\lambda_{min}]z\\
 & \geq z+2c,\lambda_{max}y+(1-\lambda_{max})z+c
\end{align*}
Therefore LP problem is

\[
\min_{x,y,z}\alpha x+\beta y+\gamma z
\]
subject to conditions above. Here as before, $\text{ (\ensuremath{\alpha},\ensuremath{\beta},\ensuremath{\gamma})=(\ensuremath{\rho\lambda_{min}},\ensuremath{\lambda_{1}}-2\ensuremath{\rho\lambda_{min}}+\ensuremath{\lambda_{2}},1-\ensuremath{\lambda_{1}}-\ensuremath{\lambda_{2}}+\ensuremath{\rho\lambda_{min}})}.$
With similar argument as in the main case, we have 
\begin{align*}
 & \alpha x+\beta y+\gamma z=\max\{z+2c,\lambda_{max}y+(1-\lambda_{max})z+c\}\\
 & z=\max\{p_{1}p_{2}x+[p_{1}(1-p_{2})+p_{2}(1-p_{1})]z,p_{1}y\}
\end{align*}
Based on which terms becomes maximum on the RHS of the first equality,
I consider the following two cases
\begin{enumerate}
\item First consider the case $z\leq\frac{p_{1}c}{\lambda_{max}(1-p_{1})}$.
In this case the maximum payout will be $z+2c$ so we have 
\[
\alpha x+\beta y+\gamma z=z+2c
\]
 by this 
\[
x=\frac{z(1-\gamma)-\beta y+2c}{\alpha}
\]
 Now divide this case to two sub-cases. 
\begin{itemize}
\item $p_{1}y=\max\{p_{1}p_{2}x+[p_{1}(1-p_{2})+p_{2}(1-p_{1})]z,p_{1}y\}=z.$
As the result 
\[
x=\frac{z(1-\gamma-\frac{\beta}{p_{1}})+2c}{\alpha}
\]
 hence the problem for LP can be written as 
\begin{align*}
 & \min z\\
 & x\geq y=\frac{z}{p_{1}}\\
 & y\geq p_{1}x+(1-p_{1})z\\
 & z=p_{1}y\geq p_{1}p_{2}x+[p_{1}(1-p_{2})+p_{2}(1-p_{1})]z
\end{align*}
 The last two inequalities can be written as 
\begin{align*}
 & z\geq\frac{p_{1}^{2}}{1-p_{1}+p_{1}^{2}}x=\text{\ensuremath{\theta_{1}x}}\\
 & z\geq\frac{p_{1}p_{2}}{1-p_{1}+2p_{1}p_{2}-p_{2}}x=\theta_{2}x
\end{align*}
 RHS of the second inequality is increasing in $p_{2}.$ So if $p_{2}\leq p^{*},$
for some $p^{*}\geq0,$ the first inequality is effective otherwise
the second one is. In any case, we get 
\[
\alpha z/\theta_{i}\geq z(1-\gamma-\frac{\beta}{p_{1}})+2c
\]
 so for optimal $z$ we have 
\[
z=\frac{2c}{\beta\frac{1-p_{1}}{p_{1}}+\alpha\frac{1-\theta_{max}}{\theta_{max}}}
\]
 which is similar to the formula as in Proposition \ref{prop: Optimal-whole-portfolio}.
When comparing with the initial inequality, we get 
\[
\alpha\frac{1-\theta_{i}}{\theta_{i}}+\beta\frac{1-p_{1}}{p_{1}}\geq\frac{2\lambda_{max}(1-p_{1})}{p_{1}}
\]
 so similar to the main case, if we have 
\[
\alpha\frac{1-\theta_{i}}{\theta_{i}}+\beta\frac{1-p_{1}}{p_{1}}
\]
 is increasing in $\rho,$ then we have shown the proposition. The
derivative with respect to $\rho$ is 
\[
\lambda_{min}(\frac{1-\theta_{i}}{\theta_{i}}-2\frac{1-p_{1}}{p_{1}})
\]
Maximum value for $\theta_{i}$ is when $p_{1}=p_{2}$ and for this
case by $p_{1}\leq\frac{1}{4}$ we get $\frac{1-\theta_{i}}{\theta_{i}}\geq2\frac{1-p_{1}}{p_{1}}$
hence the derivative above is always positive. This shows that payout
to LP is increasing in $\rho$ in this case. 
\item $p_{1}p_{2}x+[p_{1}(1-p_{2})+p_{2}(1-p_{1})]z=\max\{p_{1}p_{2}x+[p_{1}(1-p_{2})+p_{2}(1-p_{1})]z,p_{1}y\}=z$.
From this we get 
\[
x=\frac{1-[p_{1}(1-p_{2})+p_{2}(1-p_{1})]}{p_{1}p_{2}}z=qz
\]
 on the other hand, we have 
\[
x=\frac{z(1-\gamma)-\beta y+2c}{\alpha}
\]
 so we have 
\[
\beta y=2c-\alpha qz+z(1-\gamma)
\]
 which gives $y$ in terms of $z$. So LP problem is 
\[
\min z
\]
 where $p_{1}y\leq z$ and $y\geq p_{1}x+(1-p_{1})z=(p_{1}q+1-p_{1})z.$
If the coefficient in the later inequality is bigger than $\frac{1}{p_{1}},$there
is no possible solution. Otherwise, when insert $y$ from equality
above in the inequality $p_{1}y\leq z,$ in the optimal it binds hence
$p_{1}y=z$ and problem reduces to the previous case. 
\end{itemize}
\item Now consider the case $z\geq\frac{p_{1}c}{\lambda_{max}(1-p_{1})}$.
As saw above this happens for small $\rho.$ In this case we have
\begin{align*}
 & \alpha x+\beta y+\gamma z=\max\{z+2c,\lambda_{max}y+(1-\lambda_{max})z+c\}\\
 & =\lambda_{max}y+(1-\lambda_{max})z+c
\end{align*}
 like in the previous case, we consider two sub-cases. 
\begin{itemize}
\item $p_{1}y=\max\{p_{1}p_{2}x+[p_{1}(1-p_{2})+p_{2}(1-p_{1})]z,p_{1}y\}=z.$
Then the equation for expected payout can be written as 
\[
x=\frac{(\frac{\lambda_{max}}{p_{1}}+1-\lambda_{max}-\frac{\beta}{p_{1}}-\gamma)z+c}{\alpha}
\]
 similar to the previous case, LP problem becomes 
\begin{align*}
 & \min z\\
 & x\geq y=\frac{z}{p_{1}}\\
 & y\geq p_{1}x+(1-p_{1})z\\
 & z=p_{1}y\geq p_{1}p_{2}x+[p_{1}(1-p_{2})+p_{2}(1-p_{1})]z
\end{align*}
 $p_{1}x\geq z$ implies that 
\[
p_{1}c\geq(\alpha+\beta-\lambda_{max})(1-p_{1})z
\]
 which in turns implies that 
\[
z\leq\frac{p_{1}c}{(1-p_{1})\lambda_{min}(1-\rho)}
\]
 similar to the previous case, we can write the last two inequalities
as 
\begin{align*}
 & z\geq\frac{p_{1}^{2}}{1-p_{1}+p_{1}^{2}}x=\text{\ensuremath{\theta_{1}x}}\\
 & z\geq\frac{p_{1}p_{2}}{1-p_{1}+2p_{1}p_{2}-p_{2}}x=\theta_{2}x
\end{align*}
 In any case this gives us 
\[
\alpha z/\theta_{i}\geq(\frac{\lambda_{max}}{p_{1}}+1-\lambda_{max}-\frac{\beta}{p_{1}}-\gamma)z+c
\]
which gives us the optimal $z$ as the bigger term of two inequalities
above is binding. Hence we have 
\[
z=\frac{c}{\frac{\alpha(1-\theta_{max})}{\theta_{max}}+\beta\frac{1-p_{1}}{p_{1}}-\lambda_{max}\frac{1-p_{1}}{p_{1}}}
\]
 Again this is similar to the formula we have as in Proposition \ref{prop: Optimal-whole-portfolio}
and with the same reason as above $z$ is decreasing in $\rho$ which
shows our claim in this case.
\item $p_{1}p_{2}x+[p_{1}(1-p_{2})+p_{2}(1-p_{1})]z=\max\{p_{1}p_{2}x+[p_{1}(1-p_{2})+p_{2}(1-p_{1})]z,p_{1}y\}=z.$
As previous case, this gives 
\[
x=\frac{1-[p_{1}(1-p_{2})+p_{2}(1-p_{1})]}{p_{1}p_{2}}z=qz
\]
 On the other hand we have 
\[
\alpha x+\beta y+\gamma z=\lambda_{max}y+(1-\lambda_{max})z+c
\]
 which gives 
\[
\text{}(\beta-\lambda_{max})y=(1-\lambda_{max}-\gamma-\alpha q)z+c
\]
 so LP problem is 
\[
\min z
\]
 $p_{1}y\leq z$ and $y\geq p_{1}x+(1-p_{1})z=(p_{1}q+1-p_{1})z.$
Similar to the case we studied before either this does not have solution
or we get $p_{1}y=z$ in the optimum hence it reduces to the previous
case. This finishes the argument.
\end{itemize}
\end{enumerate}
\end{proof}

\subsection{Return Distribution}

Here I want to generalize the distribution function for the projects.
Let's assume the support of both types of projects $G$ and $B$ are
${0,R_{1},R_{2}}.$ For the $G$ type the chances are $\{0,p,1-p\}$
respectively and for the \textbf{$B$ }type it is $\{1-p_{1}-p_{2},p_{1},p_{2}\}.$
All the other variables, definitions and assumptions are the same
as in the main model. First want to see how contract on one project
is written. In order to persuade the optimal investment strategy (not
investing on bad project), we should have 
\begin{equation}
s_{GP}(I)\geq p_{1}s_{GP}(R_{1})+p_{2}s_{GP}(R_{2})\label{eq:fee inequality}
\end{equation}
 Also in order to motivate effort, we have 
\begin{equation}
E[s_{GP}(G)]=ps_{GP}(R_{1})+(1-p)s_{GP}(R_{2})\geq s_{GP}(I)+\frac{c}{\lambda}\label{eq:pay inequality}
\end{equation}
 In the optimal, with the same reasoning as in the binary case, both
these inequalities are binding to minimize the expected payout to
GP. LP problem is 
\[
\min_{s_{GP}(R_{1}),s_{GP}(R_{2})}\lambda E[s_{GP}(G)]+(1-\lambda)s_{GP}(I)
\]
 Conditioned to equations (equities in optimum) (\ref{eq:fee inequality})
and (\ref{eq:pay inequality}). Because of optimality, LP problem
can be written as 
\[
\min_{x,y}[\lambda p+(1-\lambda)p_{1}]x+[\lambda(1-p)+(1-\lambda)p_{2}]y
\]
 where $(x,y)=(s_{GP}(R_{1}),s_{GP}(R_{2}))$. The relation between
$x$ and $y$ comes from (\ref{eq:pay inequality}) above which can
be written as $x=\gamma y+\zeta$ where

\begin{align}
 & \gamma=-\frac{1-p-p_{2}}{p-p_{1}}\nonumber \\
 & \zeta=\frac{c}{\lambda}\label{eq:coefficients}
\end{align}
 So the minimization problem is linear in $y$ and hence in the optimum
either we have $x=0$ or $y=0$ as none of payouts can be negative.
More precisely, the coefficient of $y$ in the LP problem is 
\[
-[\lambda p+(1-\lambda)p_{1}]\frac{1-p-p_{2}}{p-p_{1}}+[\lambda(1-p)+(1-\lambda)p_{2}]
\]
 if this coefficient is positive then we should have $y=0$ otherwise
$x=0.$ Whichever happens, we get the value of the other variable
from equation $x=\gamma y+\zeta$ above. If the payout of the projects
has $n$ different values, the same conclusion holds as the problem
is linear in payouts. In summary we have 

\begin{proposition}

With setup as above, the optimal contract satisfies either $s(R_{1})=0$
or $s(R_{2})=0.$ In particular $s(R_{2})=0$ if and only if 
\[
[\lambda(1-p)+(1-\lambda)p_{2}](p-p_{1})>[\lambda p+(1-\lambda)p_{1}](1-p-p_{2})
\]
 The other one is computed by the equation $x=\gamma y+\zeta,$ where
$\gamma,\zeta$ are given in equations (\ref{eq:coefficients}). Finally
$s(I)$ is computed from the equation (\ref{eq:fee inequality}) when
it is equality.

\end{proposition}

Let's just explain the intuition behind the property that either $s(R_{1})=0$
or $s(R_{2})=0.$ LP wants to minimize the incentive for the GP to
invest in a bad project. To do this, LP considers the (weighted) difference
of probabilities that either the outcome is $R_{1}$ or $R_{2}$.
Then he makes no payment in the state with lower chance. The other
outcome associates to higher chance of investing in the good project
so GP motivates it in the contract.

Now I look at the whole-portfolio problem. Again we assume that the
policy implemented is optimal so we have $s(I)=s(R_{1})=s(R_{2})=0.$
Possible returns from optimal strategy are $2I$ when two projects
are bad, $R_{i}+I$ when one is good and other is bad and finally
$2R_{i}$ or $R_{1}+R_{2}$ when both are good. I assume parameters
for the first projects are $\lambda_{1},p,p_{1}$ and $p_{2}$. For
the second one $\lambda_{2},q,q_{1}$ and $q_{2}.$ I show the payouts
to GP by $z=s(2I),y_{i}=s(R_{i}+I),x_{i}=s(2R_{i})$ and $x=s(R_{1}+R_{2}).$
I assume correlation $\rho$ between good projects as in the main
case and assume that for good projects the realization of the returns
are independent. We have 
\begin{align}
 & s(2I)\geq\sum p_{i}s(R_{i}+I),\sum q_{i}s(R_{i}+I),\sum p_{i}q_{j}s(R_{i}+R_{j})\nonumber \\
 & ps(R_{1}+I)+(1-p)s(R_{2}+I)=E_{G_{1}}[s(R_{i}+I)]\geq\nonumber \\
 & s(2I),[pq_{2}+(1-p)q_{1}]s(R_{1}+R_{2})+pq_{1}s(2R_{1})+(1-p)q_{2}s(2R_{2})\nonumber \\
 & qs(R_{1}+I)+(1-q)s(R_{2}+I)=E_{G_{2}}[s(R_{i}+I)]\geq\nonumber \\
 & s(2I),[qp_{2}+(1-q)p_{1}]s(R_{1}+R_{2})+qp_{1}s(2R_{1})+(1-q)p_{2}s(2R_{2})\nonumber \\
 & [p(1-q)+q(1-p)]s(R_{1}+R_{2})+pqs(2R_{1})+(1-p)(1-q)s(2R_{2})\nonumber \\
 & \geq E_{G_{j}}[s(R_{i}+I)]\label{eq:motivation-investment}
\end{align}
 And finally LP should impose the equation which motivates effort.
This can be written as 
\begin{align*}
 & \rho\lambda_{min}E_{G_{1},G_{2}}[s(R_{i}+R_{j})]+(\lambda_{max}-\rho\lambda_{min})E_{G_{max}}[s(R_{i}+I)]\\
 & +(1-\rho)\lambda_{min}E_{G_{min}}[s(R_{i}+I)]+(1-\lambda_{1}-\lambda_{2}+\rho\lambda_{min})s(2I)\\
 & \geq z+2I,\lambda_{max}E_{G_{max}}[s(R_{i}+I)]+(1-\lambda_{max})z+c,\lambda_{min}E_{G_{min}}[s(R_{i}+I)]+(1-\lambda_{min})z+c
\end{align*}

As in the binary case, LP wants to minimize the expected payout to
GP (LHS of the last inequality) given constraints above. Similar to
the binary case, we can see that $E_{G_{1},G_{2}}[s(R_{i}+R_{j})]$
can be represented with an inequality. So LP problem can be written
as 
\[
\min E_{G_{1},G_{2}}[s(R_{i}+R_{j})]
\]
 subject to conditions for optimal investment and motivation for effort.
Since the general problem seems hard to solve and get good intuition
from, I stick to two important especial cases.
\begin{enumerate}
\item Suppose either $p_{1}$ and $q_{1}$ are small or $p_{2}$ and $q_{2}$
are small. In this case, suppose LP changes compensations for GP in
the case of two successful investment, while expected payout remains
the same. By this I mean changing $x_{i}$ and $x$ such that 
\[
[p(1-q)+q(1-p)]x+pqx_{1}+(1-p)(1-q)x_{2}
\]
remains fixed. By this change, RHS of the equations for optimal investment
in case of one success or no success can be changed (The first three
equation in the set of equations (\ref{eq:motivation-investment}))
. As long as RHS becomes smaller in these equations, the change can
be good (or have no effect if conditions are not binding on them).
So it is better for GP to consider payouts to minimize RHS of motivating
equations. So in this case, in the optimal contract, we get only $x_{i}>0$
(and $x_{j},x$ are zero) when $p_{i}$ and $q_{i}$ are small. Similar
reasoning implies only $y_{i}>0$ and $y_{j}=0$ when $p_{i}$ and
$q_{i}$ are small. Hence in this case problem reduces effectively
to the binary case.
\item Now consider the orthogonal problem to what we discussed in the previous
part. So I assume $p_{2}=q_{1}=0$ so the first bad project only have
return $R_{1}$ and the second one only $R_{2}$. Equations (\ref{eq:motivation-investment})
are reduced to 
\begin{align*}
 & s(2I)\geq p_{1}s(R_{1}+I),q_{2}s(R_{2}+I),p_{1}q_{2}s(R_{1}+R_{2})\\
 & ps(R_{1}+I)+(1-p)s(R_{2}+I)=E_{G_{1}}[s(R_{i}+I)]\geq\\
 & s(2I),pq_{2}s(R_{1}+R_{2})+(1-p)q_{2}s(2R_{2})\\
 & qs(R_{1}+I)+(1-q)s(R_{2}+I)=E_{G_{2}}[s(R_{i}+I)]\geq\\
 & s(2I),(1-q)p_{1}s(R_{1}+R_{2})+qp_{1}s(2R_{1})\\
 & [p(1-q)+q(1-p)]s(R_{1}+R_{2})+pqs(2R_{1})+(1-p)(1-q)s(2R_{2})\\
 & \geq E_{G_{j}}[s(R_{i}+I)]
\end{align*}
 In this case if we get $s(2I)=z,$ then similar reasoning as in binary
case, 
\begin{align*}
 & z=p_{1}s(R_{1}+I)=p_{1}^{2}s(2R_{1})\\
 & z=q_{2}s(R_{2}+I)=q_{2}^{2}s(2R_{2})\\
 & z=p_{1}q_{2}s(R_{1}+R_{2})
\end{align*}
 which implies that , as in the main case, $E[s_{GP}]$ is decreasing
in $\rho$ as well. 
\end{enumerate}


\begin{thebibliography}{10}
\bibitem{key-5}Axelson, Ulf, Per Stromberg, and Michael S. Weisbach,
2009. Why Are Buyouts Levered? The Financial Structure of Private
Equity Funds. \emph{The Journal of Finance }LXIV, 1549--1582.

\bibitem{key-1}Chung, Ji-Woong, Berk A. Sensoy, Lea Stern, and Michael
S. Weisbach, 2012. Pay for Performance from Future Fund Flows: The
Case of Private Equity. \emph{The Review of Financial Studies} 25:3259-3304.

\bibitem{key-1}Cumming Douglas and Sofia Johan, 2009. \emph{Venture
Capital and Private Equity Contracting,} Elsevier. 

\bibitem{key-1} Diamond, Douglas, 1984. Financial Intermediation
and Delegated Monitoring,\textquotedblright{} \emph{Review of Economic
Studies}, 51, 393--414.

\bibitem{key-1}Fang, Dawei, 2019. Dry powder and short fuses: Private
equity funds in emerging markets.\emph{ Journal of Corporate Finance
}59, 48-71.

\bibitem{key-6} Fulghieri, Paolo and Merih Sevilir, 2009. Size and
Focus of a Venture Capitalist\textquoteright s Portfolio. \emph{The
Review of Financial Studies} 22, 4644--4680.

\bibitem{key-7}Gompers, Paul and Josh Lerner, 1999. An analysis of
compensation in the U.S. venture capital partnership. \emph{Journal
of Financial Economics }51, 3-44.

\bibitem{key-1}He, Zhiguo, Bin Wei, Jianfeng Yu, Feng Gao, 2017.
Optimal Long-Term Contracting with Learning. \emph{The Review of Financial
Studies }30, 2006--2065.

\bibitem{key-1} Hüther, Niklas, David T. Robinson, Sönke Sievers,
and Thomas Hartmann-Wendels, 2019. Paying for Performance in Private
Equity: Evidence from Venture Capital Partnerships. \emph{Management
Science, }forthcoming.

\bibitem{key-8}Inderst, Roman, Holger M. Mueller, and Felix M\textasciidieresis unnich,
2007. Financing a Portfolio of Projects. \emph{The Review of Financial
Studies} 20, 1289--1325.

\bibitem{key-9}Koskinen, Yrjo\textasciidieresis , Michael J. Rebello,
and Jun Wang, 2014. Private Information and Bargaining Power in Venture
Capital Financing.\emph{ Journal of Economics and Management Strategy}
23, 743--775.

\bibitem{key-1} Manso, Gustavo, 2011. Motivating Innovation. \emph{The
Journal of Finance }LXVI, 1823-1869.

\bibitem{key-10}Magro, Joao. 2018. Deal-by-deal compensation structures
and portfolio diversification. \emph{Working paper.}

\bibitem{key-11}Metrick, Andrew and Ayako Yasuda, 2010. The Economics
of Private Equity Funds. \emph{The Review of Financial Studies} 23,
2303-2341.

\bibitem{key-2}Miao, Jianjun and Alejandro Rivera, 2016. Robust Contracts
in Continuous Time. \emph{Econometrica }84, 1405-1440.

\bibitem{key-4}Nachman, David and Thomas Noe, 1994. Optimal design
of securities under asymmetric information. \emph{The Review of Financial
Studies} 7, 1-44.

\bibitem{key-3}Pourbabaee, Farzad, 2020. Robust experimentation in
the continuous time bandit problem. \emph{Economic Theory}, 1-31.

\bibitem{key-3} Repullo, Rafael and Javier Suarez, 2004. Venture
Capital Finance: A Security Design Approach.\emph{ Review of Finance
}8, 75-108. 

\bibitem{key-12}Robinson, David and Berk A Sensoy, 2003. Do private
equity fund managers earn their fees? compensation, ownership, and
cash flow performance. \emph{The Review of Financial Studies} 26,
2760 - 2797.
\end{thebibliography}
\end{document}